\documentclass[11pt]{article}
\usepackage{color}
\usepackage{amssymb}
\usepackage{amsfonts}
\usepackage{amsmath}
\usepackage{amsthm}
\usepackage[small,bf,figurewithin=section]{caption}
\usepackage{graphicx}
\usepackage{graphics}
\usepackage{nicefrac}
\usepackage{multicol}
\usepackage{enumerate}
\usepackage{dsfont}
\usepackage{url}
\usepackage{capt-of}
\usepackage[english]{babel}

\title{Statistical inference for max-stable processes in space and time}

\author{Richard A. Davis
\thanks{Department of Statistics, Columbia University, New York, United States,\newline http://www.stat.columbia.edu/, Email: rdavis@stat.columbia.edu} 
\and Claudia Kl\"uppelberg
\thanks{Center for Mathematical Sciences, Technische Universit\"at M\"unchen,  
D-85748 Garching, Germany, http://www-m4.ma.tum.de, Email: cklu@ma.tum.de}
\and Christina Steinkohl
\thanks{Center for Mathematical Sciences and Institute for Advanced Study, Technische Universit\"at M\"unchen,\qquad \qquad \ \  D-85748 Garching, Germany, http://www-m4.ma.tum.de, Email: steinkohl@ma.tum.de}}

\newtheorem{theorem}{Theorem}[section]
\newtheorem{lemma}[theorem]{Lemma}

\newtheorem{proposition}[theorem]{Proposition}
\newtheorem{definition}[theorem]{Definition}
\newtheorem{corollary}[theorem]{Corollary}
\newtheorem{assumption}[theorem]{Assumption}
\newtheorem{condition}[theorem]{Condition}
\newtheorem{fig}[theorem]{Figure}

\theoremstyle{definition}
\newtheorem{example}[theorem]{Example}
\newtheorem{remark}[theorem]{Remark}

\newcommand{\bthe}{\begin{theorem}}
\newcommand{\ethe}{\end{theorem}}

\newcommand{\ben}{\begin{enumerate}}
\newcommand{\een}{\end{enumerate}}

\newcommand{\beq}{\begin{equation}}
\newcommand{\eeq}{\end{equation}}

\newcommand{\ble}{\begin{lemma}}
\newcommand{\ele}{\end{lemma}}

\newcommand{\bde}{\begin{definition}}
\newcommand{\ede}{\end{definition}}

\newcommand{\bco}{\begin{corollary}}
\newcommand{\eco}{\end{corollary}}

\newcommand{\bpr}{\begin{proposition}}
\newcommand{\epr}{\end{proposition}}

\newcommand{\bproof}{\begin{proof}}
\newcommand{\eproof}{\end{proof}}

\newcommand{\bexam}{\begin{example}}
\newcommand{\eexam}{\end{example}}

\newcommand{\bfi}{\begin{fig}}
\newcommand{\efi}{\end{fig}}

\newcommand{\btab}{\begin{tab}}
\newcommand{\etab}{\end{tab}}

\newcommand{\beao}{\begin{eqnarray*}}
\newcommand{\eeao}{\end{eqnarray*}\noindent}

\newcommand{\beam}{\begin{eqnarray}}
\newcommand{\eeam}{\end{eqnarray}\noindent}

\newcommand{\barr}{\begin{array}}
\newcommand{\earr}{\end{array}}

\newcommand{\bdis}{\begin{displaymath}}
\newcommand{\edis}{\end{displaymath}\noindent}
\newcommand{\bs}{\boldsymbol}

\newcommand{\bbn}{\mathbb{N}}
\newcommand{\bbz}{\mathbb{Z}}
\newcommand{\bbr}{\mathbb{R}}

\newcommand{\coloneqq}{\mathrel{\mathop:}=}
\newcommand{\eqqcolon}{=\mathrel{\mathop:}}

\DeclareMathOperator*{\argmax}{arg\,max}

\begin{document}

\maketitle

\begin{abstract}
Max-stable processes have proved to be useful for the statistical
modelling of spatial extremes. Several representations of max-stable
random fields have been proposed in the literature. 
One such representation is based on a limit of normalized
and scaled pointwise maxima of stationary Gaussian processes that was first
introduced by Kabluchko, Schlather and de Haan \cite{Schlather2}.

This paper deals with statistical inference for max-stable space-time processes that are defined in an analogous fashion. We describe pairwise likelihood estimation, where the pairwise density of the process is used to estimate the model parameters and prove strong consistency and asymptotic normality of the parameter estimates for an increasing space-time dimension, i.e., as the joint number of spatial locations and time points tends to infinity. 
A simulation study shows that the proposed method works well for these models. 
\end{abstract}
\vfill

\noindent
\begin{tabbing}
\> secondary: \,\,\,62F12, 62M10, 62M40
\end{tabbing}

\vspace{1cm}

\vspace{0.5cm}

\section{Introduction}
\label{Introduction}
Max-stable processes have proven to be useful in the modelling of spatial extremes. 
Typically, meteorological extremes like heavy rainfall or extreme wind speeds are modelled using extreme value theory. In particular, sample maxima such as annual maximum wind speeds are observed at several locations of some spatial process. Other applications may involve the analysis of image data resulting from tomographic examinations. 

Several representations of max-stable processes have been proposed in the literature, including for example Brown and Resnick \cite{Brown}, de Haan \cite{deHaan}, Kabluchko, Schlather and de Haan \cite{Schlather2}, and Schlather \cite{Schlather1}. 
Recently, models for extreme values observed in a space-time setting have generated a great deal of interest. First approaches can be found in Davis and Mikosch \cite{Davis}, Huser and Davison \cite{Huser}, Kabluchko \cite{Kabluchko1}, and Davis, Kl\"uppelberg and Steinkohl \cite{Steinkohl}. 

In this paper, we follow the approach described in Davis et al. \cite{Steinkohl}, who extend the max-stable process introduced in Kabluchko et al. \cite{Schlather2} to a space-time setting. The process is constructed as the limit of rescaled and normalized maxima of independent replications of some stationary Gaussian space-time process. The underlying correlation function of the Gaussian process is assumed to belong to a parametric model whose parameters describe smoothness of the correlation function near the origin. 

As it is well-known for max-stable processes, the full likelihood function is computationally intractable and other methods have to be used to derive parameter estimates. Standard procedures for such cases are composite likelihood including pairwise likelihood estimation. These methods go back to Besag \cite{Besag}, and there is an extensive literature available dealing with applications and properties of the estimates, see for example Cox and Reid \cite{Cox}, Lindsay \cite{Lindsay}, Varin \cite{Varin},  or Varin and Vidoni \cite{Varin2}.   
Recent work concerning the application of pairwise likelihood methods to max-stable random fields can be found in Huser and Davison \cite{Huser} and Padoan, Ribatet and Sisson \cite{Ribatet}. 

Since the observations in a space-time setting are correlated in space and time, we use special properties of max-stable processes to show strong consistency and asymptotic normality of the estimates. Here, it is assumed that the locations lie on a regular lattice and that the time points are equidistant. The spatial and/or the temporal dimension, i.e., the number of spatial locations and/or time points, increases to infinity.  
The main step of the proof is based on a strong law of large numbers for the pairwise likelihood function. Stoev \cite{Stoev2} analyzed ergodic properties for max-stable processes in time resulting from extremal integral representations for max-stable processes that were introduced in Stoev and Taqqu \cite{Stoev1}. The extension to a spatial setting and the resulting strong law of large numbers was shown by Wang, Roy and Stoev \cite{Wang}. By combining these two results we obtain a strong law of large numbers for a jointly increasing space-time domain. 

In addition to strong consistency, we prove asymptotic normality for the pairwise likelihood estimates. A first result concerning asymptotic normality of pairwise likelihood estimates for max-stable space-time processes can be found Huser and Davison \cite{Huser}, who fix the number of locations and let the number of time points tend to infinity. We formulate asymptotic normality for the space-time setting and use Bolthausen's theorem \cite{Bolthausen} together with strong mixing properties shown by Dombry and Eyi-Minko \cite{Dombry} to prove asymptotic normality for an increasing number of space-time locations. 

Our paper is organized as follows. In Section \ref{Desmodel}, we introduce the max-stable space-time process for which inference properties will be considered in subsequent sections. Section \ref{PLmodel} describes pairwise likelihood estimation and the particular setting for our model. In Section \ref{Consistency1} we establish strong consistency for the estimates for increasing space-time domain. Asymptotic normality of these parameters is established in Section \ref{Asymptoticnormality}. A simulation study evaluating the performance of the estimates is described in Section \ref{Simulation}. 

\section{Description of the model}\label{Desmodel}
We start with the process that will be used for modelling extremes in space and time; details can be found in Davis et al. \cite{Steinkohl}. 
Let $\left\{Z(\bs{s},t), \bs{s}\in \bbr^d, t\in [0,\infty)\right\}$ denote a stationary space-time Gaussian process on $\bbr^d\times [0,\infty)$ with mean zero and variance one.
With the correlation function
$$\rho(\bs{h},u) = \mathbb{E}\left[Z(\bs{s},t)Z(\bs{s}+\bs{h},t+u)\right],$$
where $\bs{h}\in\bbr^d$ is the spatial lag and $u\in\bbr$ is the time lag,   
we make the following assumption that will be used throughout the paper. 
\begin{assumption}\label{ass1}
There exist sequences of constants $s_n \to 0$, $t_n \to 0$ as $n\to \infty$, such that
\begin{equation*}
\log n(1-\rho(s_n \bs{h},t_n u)) \to \delta(\bs{h},u) >0, \ \text{ as } n\to \infty.
\end{equation*}
\end{assumption}
Assumption \ref{ass1} is natural in the context of stationary space-time models; the correlation function tends to one at a certain rate as the space-time lag approaches the zero.
\begin{proposition}[Kabluchko et al. \cite{Schlather2} and Davis et al. \cite{Steinkohl}]\label{BrownResnick}
Let $\left\{Z_j(\bs{s},t),\bs{s}\in\bbr^d,t\in[0,\infty)\right\}, j=1,\ldots,n$, be independent replications of the space-time Gaussian process described above and let $\left\{\xi_j, j\in \bbn\right\}$ denote points of a Poisson random measure on $[0,\infty)$ with intensity measure $\xi^{-2}d\xi$. 
Suppose Assumption \ref{ass1} is satisfied. Then, the random fields $\left\{\eta_n(\bs{s},t)), \bs{s}\in \bbr^d, t\in [0,\infty)\right\}$, defined for $n\in \bbn$ by
\begin{equation}
\eta_n(\bs{s},t) = \bigvee\limits_{j=1}^n -\frac{1}{\log(\Phi(Z_j(s_n\bs{s},t_nt)))}, \ \bs{s}\in \bbr^d, t\in [0,\infty),
\label{model}
\end{equation}
converge weakly on the space of continous functions on $\bbr^d\times[0,\infty)$ to the stationary Brown-Resnick process 
\begin{equation}
\eta(\bs{s},t) = \bigvee\limits_{j=1}^{\infty}\xi_j \exp\left\{W_j(\bs{s},t) - \delta(\bs{s},t)\right\},
\label{limitfield}
\end{equation}
where the deterministic function $\delta$ is given in Assumption \ref{ass1} and $\left\{W_j(\bs{s},t), \bs{s}\in \bbr^d,t\in[0,\infty)\right\}$, $j\in \bbn$ are independent replications of a Gaussian process with stationary increments, $W(\bs{0},0) = 0$, $\mathbb{E}(W(\bs{s},t)) = 0$ and covariance function for $\bs{s}_1,\bs{s}_2 \in \bbr^d, t_1,t_2 \in [0,\infty)$
$$\mathbb{C}ov\left(W(\bs{s}_1,t_1),W(\bs{s}_2,t_2)\right) = \delta(\bs{s}_1,t_1) + \delta(\bs{s}_2,t_2) - \delta(\bs{s}_1-\bs{s}_2,t_1-t_2).$$ 
The bivariate distribution function of $\eta$ can be expressed in closed form and is based on a well-known result by H\"usler and Reiss \cite{Huesler};
\begin{equation}\label{bivhuesler}
F(x_1,x_2) =\exp\left\{-\frac{1}{x_1}\Phi\left(\frac{\log\frac{x_2}{x_1}}{2\sqrt{\delta(\bs{h},u)}} + \sqrt{\delta(\bs{h},u)}\right) - \frac{1}{x_2}\Phi\left(\frac{\log\frac{x_1}{x_2}}{2\sqrt{\delta(\bs{h},u)}}+ \sqrt{\delta(\bs{h},u)}\right)\right\},
\end{equation}
where $\Phi$ denotes the distribution function of a standard normal distribution. 
\end{proposition}

Many correlation functions satisfy the following condition, which will be used throughout. 
\begin{condition}\label{ass2}
The correlation function has an expansion around zero, given by 
\begin{equation*}
\rho(\bs{h},u) = 1-\theta_1\|\bs{h}\|^{\alpha_1}-\theta_2|u|^{\alpha_2} + O(\|\bs{h}\|^{\alpha_1}|u|^{\alpha_2}), \bs{h}\in\bbr^d, u\in \bbr,
\end{equation*}
where $0<\alpha_1,\alpha_2\leq 2$ and $\theta_1,\theta_2 >0$. 
\end{condition}
Condition \ref{ass2} allows for an explicit expression of the limit function $\delta$ in Assumption \ref{ass1},
\begin{equation}\label{delta}
\delta(\bs{h},u) = \theta_1\|\bs{h}\|^{\alpha_1} + \theta_2|u|^{\alpha_2},
\end{equation} 
where the scaling sequences $(s_n)$ and $(t_n)$ can be chosen as 
$s_n = (\log n)^{1/\alpha_1}$ and $t_n = (\log n)^{1/\alpha_2}$. 
The parameters $\alpha_1,\alpha_2 \in (0,2]$ relate to the smoothness of the underlying Gaussian process in space and time, where the case $\alpha_1=\alpha_2=2$ corresponds to a mean-square differentiable process. 
For example, Gneiting's class of correlation functions \cite{Gneiting} satisfies Condition \ref{ass2}. For a detailed analysis of Gneiting's class and further examples we refer to Davis et al. \cite{Steinkohl}, Proposition 4.5, where the expansion around zero is calculated for several classes of correlation functions. 
A further property of the model defined in Proposition \ref{BrownResnick} is the closed form expression for the \emph{tail dependence coefficient}, which is defined by 
$$\chi(\bs{h},u) = \lim\limits_{x\to\infty}P\left(\eta(\bs{s}_1,t_1)>F^{\leftarrow}_{\eta(\bs{s}_1,t_1)}(x)\mid\eta(\bs{s}_2,t_2)>F^{\leftarrow}_{\eta(\bs{s}_2,t_2)}(x)\right),$$
where $\bs{h}=\bs{s}_1-\bs{s}_2$ denotes the spatial distance between two locations and $u=t_1-t_2$ is the temporal lag. 
As derived in Section 3 in Davis et al. \cite{Steinkohl}, we obtain
\begin{equation}
\chi(\bs{h},u) = 2\left(1-\Phi(\sqrt{\delta(\bs{h},u)})\right) = 2\left(1-\Phi(\sqrt{\theta_1\|\bs{h}\|^{\alpha_1} + \theta_2|u|^{\alpha_2}})\right).
\label{chi}
\end{equation}

\section{Pairwise likelihood estimation}\label{Plmodel}
In this section, we describe the pairwise likelihood estimation for the parameters of the model in \eqref{model} introduced in Section \ref{Desmodel}. 
Composite likelihood methods have been used, whenever the full likelihood is not available or intractable. We present the general definition of composite and pairwise likelihood functions in Section \ref{comp}. Afterwards, we describe the details for our model.
\subsection{Basics on composite likelihood estimation}\label{comp}
Composite likelihood methods go back to Besag \cite{Besag} and Lindsay \cite{Lindsay} and there is vast literature available, from a theoretical and an applied point of view. 
For more information we refer to Varin \cite{Varin} who gives an overview of existing models and inference including extensive references.  
In the most general setting, the composite log-likelihood function  
is given by
\begin{equation*}
l_c(\bs{\psi},\bs{x}) = \sum\limits_{i=1}^q w_i \log f(\bs{x}\in A_i; \bs{\psi}).
\end{equation*}
From this general form, special composite likelihood functions 
can be derived. For our setting we define the \emph{(weighted) pairwise log-likelihood function} by 
\begin{equation}
PL(\bs{\psi};\bs{x}) = \sum\limits_{i=1}^{n}\sum\limits_{j=1}^n w_{i,j} \log f_{\bs{\psi}}(x_i,x_j),
\label{PLgendef}
\end{equation}
where $\bs{x} = (x_1,\ldots,x_n)$ is the data vector, $f_{\bs{\psi}}(x_i,x_j)$ is the density for the bivariate observations $(x_i,x_j)$ and $w_{i,j}$ are weights which can be used for example to reduce the number of pairs included in the estimation. 
The parameter estimates are obtained by maximizing \eqref{PLgendef}. 

As noted in Cox and Reid \cite{Cox}, for dependent observations, estimates based on the composite likelihood need not be consistent or asymptotically normal. 
This is important for space-time applications, since all components may be highly dependent across space and time. 

\subsection{Application to spatio-temporal max-stable random fields}\label{PLmodel}
To derive pairwise-likelihood functions for the model defined in Proposition \eqref{model} we first need to derive the bivariate density function for the space-time max-stable process. For later purposes we state the closed form expression in the following lemma. Throughout we denote by $\Phi$ and $\varphi$ the cumulative distribution function and the density of the standard normal distribution, respectively. For simplicity we suppress the argument $(x_1,x_2)$.  
\begin{lemma}\label{bivdenlemma}
Set $\delta \coloneqq  \delta(\bs{h},u)$ as given in \eqref{delta} and define for $x_1,x_2 >0$
\begin{align}
q_{\bs{\psi}}^{(1)} &\coloneqq \frac{\log (x_2/x_1)}{2\sqrt{\delta}} + \sqrt{\delta} \qquad 
q_{\bs{\psi}}^{(2)}\coloneqq \frac{\log (x_1/x_2)}{2\sqrt{\delta}} + \sqrt{\delta}, \label{q1q2}\\
V &\coloneqq \frac{1}{x_1}\Phi(q_{\bs{\psi}}^{(1)})+\frac{1}{x_2}\Phi(q_{\bs{\psi}}^{(2)}). \label{V}
\end{align} 
The partial derivatives of $q_{\bs{\psi}}^{(1)}$ and $q_{\bs{\psi}}^{(2)}$ are given by
\begin{align*}
\frac{\partial q_{\bs{\psi}}^{(1)}}{\partial x_1} &= -\frac{1}{2\sqrt{\delta}x_1},  \quad 
\frac{\partial q_{\bs{\psi}}^{(1)}}{\partial x_2} = \frac{1}{2\sqrt{\delta}x_2}, \quad  \frac{\partial q_{\bs{\psi}}^{(2)}}{\partial x_2} = -\frac{1}{2\sqrt{\delta}x_2},  \quad 
\frac{\partial q_{\bs{\psi}}^{(2)}}{\partial x_1} = \frac{1}{2\sqrt{\delta}x_1}.\\\end{align*}
The first and second order partial derivatives of $V$ are given by
\begin{align*}
\frac{\partial V}{\partial x_1} &= -\frac{1}{x_1^2}\Phi(q_{\bs{\psi}}^{(1)})- \frac{1}{2\sqrt{\delta}x_1^2} \varphi(q_{\bs{\psi}}^{(1)}) + \frac{1}{2\sqrt{\delta}x_1x_2} \varphi(q_{\bs{\psi}}^{(2)}), \\
\frac{\partial V}{\partial x_2} &= -\frac{1}{x_2^2} \Phi(q_{\bs{\psi}}^{(2)}) - \frac{1}{2\sqrt{\delta}x_2^2} \varphi(q_{\bs{\psi}}^{(2)}) + \frac{1}{2\sqrt{\delta}x_1x_2}\varphi(q_{\bs{\psi}}^{(1)}),\\
\frac{\partial^2 V}{\partial x_1 \partial x_2}
&= -\frac{2\sqrt{\delta}-q_{\bs{\psi}}^{(1)}}{4\delta x_1^2x_2} \varphi(q_{\bs{\psi}}^{(1)}) - \frac{2\sqrt{\delta} - q_{\bs{\psi}}^{(2)}}{4\delta x_1x_2^2} \varphi(q_{\bs{\psi}}^{(2)}).
\end{align*}  
Finally, the \emph{bivariate log-density} is
\begin{equation}
\log f(x_1,x_2) = -V + \log\left[\left(\frac{\partial V}{\partial x_1}\right)\left(\frac{\partial V}{\partial x_2}\right) - \frac{\partial^2 V}{\partial x_1 \partial x_2}\right].
\label{bivden}
\end{equation}
\end{lemma}

The resulting parameter vector is  $\bs{\psi}=(\theta_1,\alpha_1,\theta_2,\alpha_2)$.
We first define the pairwise likelihood for a general setting with $M$ locations $\bs{s}_1,\ldots,\bs{s}_M$ and $T$ time points $0\leq t_1<\cdots <t_T<\infty$. In a second step we assume that the locations lie on a regular grid and that the time points are equidistant. 

\begin{align}
& PL^{(M,T)}(\bs{\psi}) =  \sum\limits_{i=1}^{M}\sum\limits_{j=i+1}^M\sum\limits_{k=1}^{T-1}\sum\limits_{l=k+1}^{T} w^{(M)}_{i,j}w^{(T)}_{k,l}\log f_{\bs{\psi}}(\eta(\bs{s_i},t_k),\eta(\bs{s}_j,t_l)),
\label{PLfull0}
\end{align}
where $w^{(M)}_{i,j}\geq 0$ and $w^{(T)}_{k,l}\geq 0$ denote spatial and temporal weights, respectively. Since it is expected that space-time pairs, which are far apart in space or in time, have only little influence on the dependence parameters to be estimated, we define the weights, such that in the estimation only pairs with a maximal spatio-temporal distance of $(r,p)$ are included, i.e.,  
\begin{equation}
w^{(M)}_{i,j} = \mathds{1}_{\left\{\|\bs{s}_i-\bs{s}_j\| \leq r\right\}}, \qquad w^{(T)}_{k,l} = \mathds{1}_{\left\{|t_k-t_l| \leq p\right\}},
\label{weights}
\end{equation}
where $\|\cdot\|$ denotes any arbitrary norm on $\bbr^d$. The pairwise likelihood estimates are given by 
\begin{equation}
(\hat{\theta}_1,\hat{\alpha}_1,\hat{\theta}_2,\hat{\alpha}_2) = \argmax\limits_{(\theta_1,\alpha_1,\theta_2,\alpha_2)} PL^{(M,T)}(\theta_1,\alpha_1,\theta_2,\alpha_2). 
\label{PLfull}
\end{equation}
Using the definition of the weights in \eqref{weights}, the log-likelihood function in \eqref{PLfull0} can be rewritten as 
\begin{align*}
& PL^{(M,T)}(\bs{\psi}) =  \sum\limits_{i=1}^{M}\sum\limits_{\stackrel{j=i+1}{\|\bs{s}_i-\bs{s}_j\|\leq r}}^{M}\sum\limits_{k=1}^{T-p}\sum\limits_{l=k+1}^{\min\left\{k+p,T\right\}} \log f_{\bs{\psi}}(\eta(\bs{s}_i,t_k),\eta(\bs{s}_j,t_l)). 
\end{align*}

The following sampling scheme is assumed throughout.  
\begin{condition}\label{grid}
We assume that the locations lie on a regular $d$-dimensional lattice, 
$$S = \big{\{}\bs{s}_j, j=1,\ldots,m^d\big{\}} = \big{\{}\bs{s}_{(i_1,\ldots,i_d)}=(i_1,\ldots,i_d),\ i_1,\ldots,i_d \in \left\{1,\ldots,m\right\}\big{\}}.$$
Further assume that the time points are equidistant,
$$0 = t_0\leq t_1<\cdots < t_T<\infty, \ \ |t_k-t_{k-1}|=1,\  k=1,\ldots, T.$$
\end{condition}

For later purposes, we rewrite the pairwise log-likelihood function under Condition \ref{grid} in the following way. 
Define $\mathcal{H}_r$ as the set of all vectors with non-negative integer-valued components $\bs{h}$ without the $\bs{0}$-vector, which point to other sites in the set of locations within distance $r$. Nott and Ryden \cite{Nott} call this the \emph{design mask}. We denote by $|\mathcal{H}_r|$ the cardinality of the set $\mathcal{H}_r$. In our application, we will use the following design masks according to the Euclidean distance; 
\begin{align*}
\mathcal{H}_1 &= \left\{(1,0),(0,1)\right\}\\
\mathcal{H}_2 &= \mathcal{H}_1 \cup \left\{(1,1),(0,2),(2,0)\right\} \\
\mathcal{H}_3 &= \mathcal{H}_2 \cup \left\{(1,2),(2,1),(2,2),(0,3),(3,0)\right\} \\
\mathcal{H}_4 &= \mathcal{H}_3 \cup \left\{(1,3),(3,1),(2,3),(3,2),(4,0),(0,4)\right\}\\
 &\vdots \\
\end{align*}
Using Condition \ref{grid} and the design mask, the \emph{pairwise log-likelihood function} in \eqref{PLfull} can be rewritten as  
\begin{align}\label{PLfull2}
& PL^{(m^d,T)}(\bs{\psi})\nonumber \\
& =\sum\limits_{i_1=1}^{m}\cdots\sum\limits_{i_d=1}^{m}\sum\limits_{\stackrel{\bs{h}\in \mathcal{H}_r}{\bs{s}_{(i_1,\ldots,i_d)} + \bs{h} \in S}}\sum\limits_{k=1}^{T}\sum\limits_{l=k+1}^{\min\left\{k+p,T\right\}} \log f_{\bs{\psi}}(\eta(\bs{s}_{(i_1,\ldots,i_d)},t_k),\eta(\bs{s}_{(i_1,\ldots,i_d)}+\bs{h},t_l)) \nonumber \\
&= \sum\limits_{i_1=1}^{m}\cdots\sum\limits_{i_d=1}^{m}\sum\limits_{k=1}^{T} g_{\bs{\psi}}\left(i_1,\ldots,i_d,k;\mathcal{H}_r,p\right) - \mathcal{R}^{(m^d,T)}(\bs{\psi}), 
\end{align}
where 
\begin{align}
g_{\bs{\psi}}(i_1,\ldots,i_d,k;\mathcal{H}_r,p) &= 
\sum\limits_{\bs{h}\in\mathcal{H}_r}\sum\limits_{l=k+1}^{k+p}\log f_{\bs{\psi}}(\eta(\bs{s}_{(i_1,\ldots,i_d)},t_k),\eta(\bs{s}_{(i_1,\ldots,i_d)}+\bs{h},t_l)), \label{gpsi}
\end{align}
and 
$\mathcal{R}^{(m^d,T)}(\bs{\psi})$ is a boundary term, given by
\begin{align}
\mathcal{R}^{(m^d,T)}(\bs{\psi})&= \sum\limits_{i_1=1}^{m}\cdots\sum\limits_{i_d=1}^m\sum\limits_{\stackrel{\bs{h}\in \mathcal{H}_r}{\bs{s}_{(i_1,\ldots,i_d)} + \bs{h}\notin S}}\sum\limits_{k=1}^{T}\sum\limits_{\stackrel{l=k+1}{l>T}}^{k+p} \log f_{\bs{\psi}}(\eta(\bs{s}_{(i_1,\ldots,i_d)},t_k),\eta(\bs{s}_{(i_1,\ldots,i_d)}+\bs{h},t_l))
\label{boundary}
\end{align}

\begin{figure}[ht]
\centering
\vspace*{-1cm}
\includegraphics[scale=0.8]{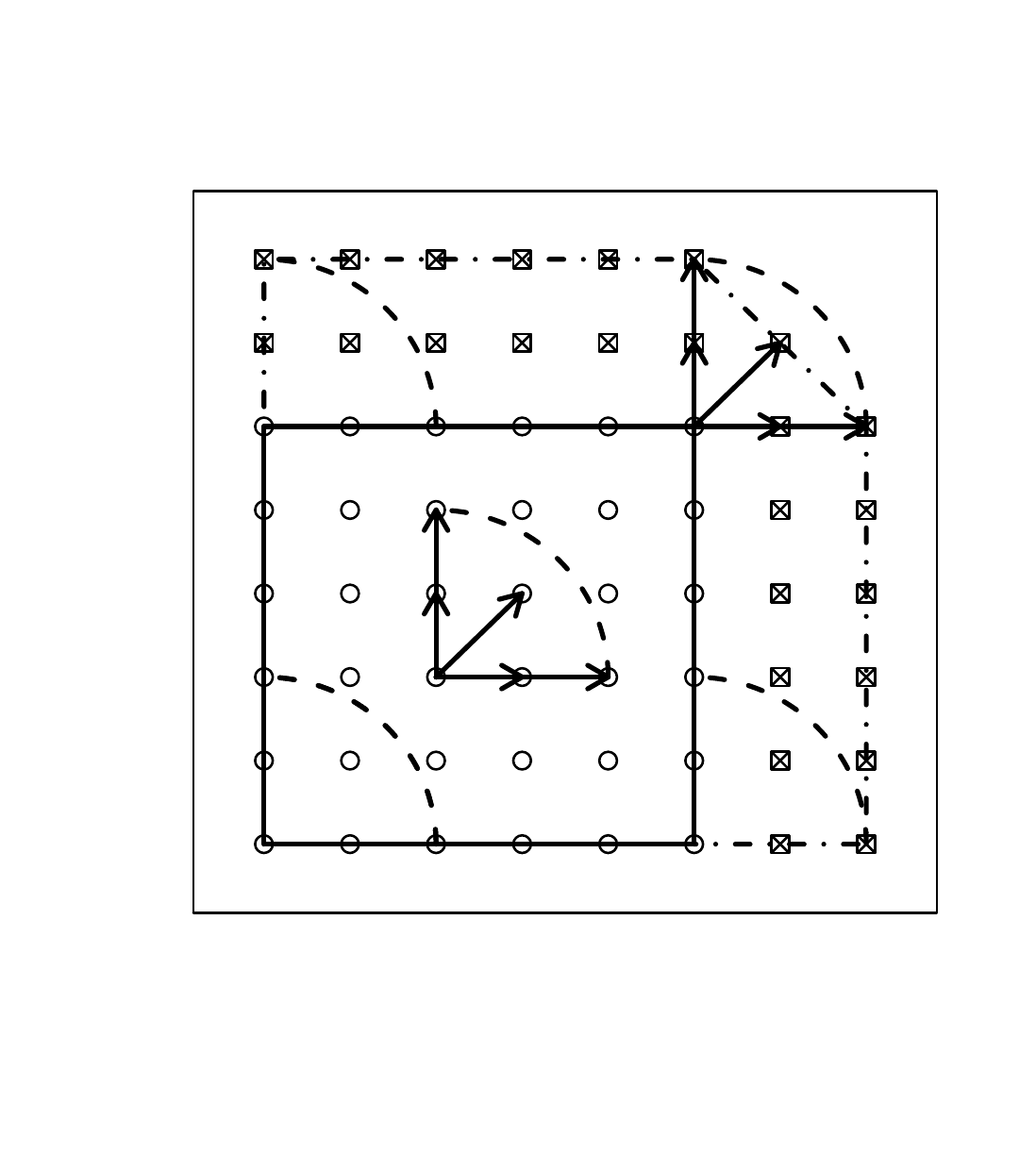}
\vspace*{-1.5cm}
\caption{Visualization of the boundary term $\mathcal{R}^{(m^d,T)}$ for $d=2$, $m=6$ and any time point, where the set $S$ of locations is the inner square and the outer polygon represents the endpoints of pairs in the boundary}
\label{gridskizze}
\end{figure}
Figure \ref{gridskizze} depcits a spatial grid with length $m=6$, where the inner square is the set of observed locations $S$ and the points in the outer polygon are endpoints of pairs which are in the boundary term $\mathcal{R}^{(m^d,T)}$. The figure visualizes the case $\mathcal{H}_2$ which is represented by the quarter circles.  

\section{Strong consistency of the pairwise likelihood estimates}\label{Consistency1}
In this section we establish strong consistency for the pairwise likelihood estimates introduced in Section \ref{PLmodel}. 
For univariate time series models Davis and Yau \cite{Davis5} proved strong consistency of the composite likelihood estimates in full detail.  
For max-stable random fields with replicates, which are independent in time, Padoan et al. \cite{Ribatet} showed consistency and asymptotic normality for the pairwise likelihood estimates. 
In contrast to previous studies, where either the spatial or the time domain increases,  we show strong consistency as the space-time domain increases jointly. 



\subsection{Ergodic properties for max-stable processes}\label{Mixing}

Stoev and Taqqu \cite{Stoev1} introduced extremal integrals as an analogy to sum-stable integrals. Based on the extremal integral representation of max-stable processes Stoev \cite{Stoev2} establishes conditions under which the max-stable process is ergodic. Wang et al. \cite{Wang} extend these results to a spatial setting. 
In the following, let $\tau_{(h_1,\ldots,h_d,u)}$ denote the multiparameter shift-operator.
In accordance with the definitions and results in Wang et al. \cite{Wang}, we define ergodic and mixing space-time processes. 
\begin{definition}
A strictly stationary space-time process $\left\{\eta(\bs{s},t), \bs{s}\in \bbr^d, t\in [0,\infty)\right\}$ is called \emph{ergodic}, if for all $A,B \in \sigma\left\{\eta(\bs{s},t),\bs{s}\in\bbr^d, t\in [0,\infty)\right\}$
\begin{equation}
\lim\limits_{m_1\cdots m_d T\to \infty}\frac{1}{m_1\cdots m_d T}\sum\limits_{h_1=1}^{m_1}\cdots \sum\limits_{h_d=1}^{m_d} \sum\limits_{u=1}^T P\left(A\cap\tau_{(h_1,\ldots,h_d,u)}(B)\right) = P(A)P(B).
\label{ergodic}
\end{equation} 
If the process satisfies additionally
\begin{equation}
\lim\limits_{n\to \infty}P\left(A\cap\tau_{(s_{1,n},\ldots,s_{d,n},t_n)}(B)\right) = P(A)P(B),
\end{equation}
for all sequences $\left\{(s_{1,n},\ldots,s_{d,n},t_n), n\in \bbn\right\}$ with $\max\left\{|s_{1,n}|,\ldots,|s_{d,n}|,|t_n|\right\}\to \infty$, we call the process \emph{mixing}. 
\end{definition}
Note in \eqref{ergodic} that in contrast to the ergodic theorem in Wang et al. \cite{Wang}, the number of terms in each sum is not equal, since we have an additional sum for the time component. 
We focus on max-stable processes with extremal integral representation
\begin{equation}
\eta(s_1,\ldots,s_d,t) = \stackrel{e}{\int\limits_{E}} U_{(s_1,\ldots,s_d,t)}(f) dM_{1}, 
\label{extrep}
\end{equation}
where $U_{(s_1,\ldots,s_d,t)} : L^1(\mu) \to L^1(\mu)$ given by $U_{(s_1,\ldots,s_d,t)}(f) = f\circ \tau_{(s_1,\ldots,s_d,t)}$ is a group of max-linear automorphisms with $U_{(0,\ldots,0,0)}(f) = f$ and the control measure $\mu$ is the distribution of the space-time process, 
$$\mu(A) = P\left(\eta(s_1,\ldots,s_d,t) \in A\right), \quad A \in \sigma\left\{\eta(\bs{s},t),(\bs{s},t)\in\bbr^d\times [0,\infty)\right\}.$$ 
The following result is a direct extension of the uniparameter theorem established in Stoev \cite{Stoev2}, Theorem 3.4, 
and the multiparameter counterpart in Wang et al. \cite{Wang}. 
\begin{proposition}[Wang et al. \cite{Wang}, Theorem 5.6]
The max-stable process defined in \eqref{extrep} is mixing, if and only if 
\begin{equation}
\int_E U_{(s_{1,n},\ldots,s_{d,n},t_n)}(f) \wedge U_{(0,\ldots,0,0)}(f) d\mu = \int_E U_{(s_{1,n},\ldots,s_{d,n},t_n)}(f) \wedge f d\mu \to 0,
\end{equation}
for all sequences $\left\{(s_{1,n},\ldots,s_{d,n},t_n)\right\}$ with $\max\left\{|s_{1,n}|,\ldots,|s_{d,n}|,|t_n|\right\}\to \infty$ as $n\to\infty$. 
\end{proposition}

Wang et al. \cite{Wang} showed, that the ergodic theorem stated above holds for mixing max-stable processes with extremal integral representation \eqref{extrep} in the case of $T=m$. The extension to the multiparameter case where $T\neq m$ is a simple generalisation using Theorem 6.1.2 in Krengel \cite{Krengel}, which is a multiparameter extension of the Akcoglu's ergodic theorem. Ergodic properties of Brown-Resnick processes have been studied for the uniparameter case in Stoev and Taqqu \cite{Stoev1} and Wang and Stoev \cite{Wang2}. We summarize the results in the following proposition. 

\begin{proposition}\label{mixing1}
The Brown-Resnick process in Proposition \ref{BrownResnick} with extremal integral representation 
$$\bigg{\{}\stackrel{e}{\int\limits_{E}} \exp\left\{W(\bs{s},t)-\delta(\bs{s},t)\right\}dM_1\ \ \bs{s}\in \bbr^d, t\in [0,\infty)\bigg{\}}$$
is mixing in space and time. 
The strong law of large numbers holds;
\begin{equation}
\frac{1}{m^dT}\sum\limits_{i_1=1}^m\cdots\sum\limits_{i_d=1}^m\sum\limits_{k=1}^T g(\eta(\bs{s}_{(i_1,\ldots,i_d)},t_k)) \stackrel{a.s.}{\rightarrow} \mathbb{E}\left[g(\eta(\bs{s}_{(1,\ldots,1)},t_1))\right], \quad mT\to \infty
\label{SLLN}
\end{equation}
where $g$ is a measurable function. 
\end{proposition}

\subsection{Consistency for large $mT$}\label{Consistency}
In the following we show that the pairwise likelihood estimate resulting from maximizing \eqref{PLfull2} for the model defined in Proposition \ref{BrownResnick} is strongly consistent.
\begin{theorem}\label{Consistencythm}
Assume that the true parameter vector $\bs{\psi}^* = (\theta_1^*,\alpha_1^*,\theta_2^*,\alpha_2^*)$ lies in a compact set $\Psi$, which does not contain $\bs{0}$ and which satisfies for some $c>0$ 
\begin{equation}
\Psi \subseteq\left\{\min\left\{\theta_1,\theta_2\right\}>c, \alpha_1,\alpha_2\in (0,2]\right\}.
\label{parspace}
\end{equation}
Assume also that the identifiability condition 
\begin{align}
\bs{\psi}  = \bs{\widetilde{\psi}}&\quad \Leftrightarrow \quad   
f_{\bs{\psi}}(\eta(\bs{s}_1,t_1),\eta(\bs{s}_2,t_2)) = f_{\bs{\widetilde{\psi}}}(\eta(\bs{s}_1,t_1),\eta(\bs{s}_2,t_2)), \ \ \text{a.s.}
\label{ident1}
\end{align}
is satisfied for all $(\bs{s}_1,t_1),(\bs{s}_2,t_2)$. 
It then follows that the pairwise likelihood estimate 
$$\bs{\hat{\psi}}_{m^d,T} = \argmax\limits_{\bs{\psi} \in \Psi} PL^{(m^d,T)}(\bs{\psi})$$
is strongly consistent, i.e. $\bs{\hat{\psi}}_{m^d,T} \stackrel{a.s.}{\to} \bs{\psi}^*$ as $mT \to \infty$. 
\end{theorem}

\begin{remark}\label{Identifiability}
For the identifiability assumption \eqref{ident1} we consider different cases according to the maximal space-time lag $(r,p)$ included in the composite likelihood. Recall that the pairwise density, see Lemma \ref{bivdenlemma}, depends on the spatial distance $\bs{h}$ and the time lag $u$ only through the function $\delta(\bs{h},u) = \theta_1\|\bs{h}\|^{\alpha_1}+\theta_2|u|^{\alpha_2}$. For specific combinations of $(r,p)$ not all parameters are identifiable. 
For example, if the maximal spatial lag taken into account in the estimation equals one, i.e. $r=1$, and $p>1$, the parameter $\alpha_1$ is not identifiable. Strong consistency still holds for the remaining parameters. Table \ref{identpar} lists the various scenarios. 
\begin{table}[h]
\centering 
\begin{tabular}{l|l|l}
Maximal spatial lag $r$& Maximal temporal lag $p$& Identifiable parameters \\ 
\hline
0 & 1 & $\theta_2$ \\
0 & $p$, $p>1$ & $\theta_2$, $\alpha_2$ \\
1 & 0 & $\theta_1$ \\
$r$, $r>1$ & 0 & $\theta_1$, $\alpha_1$ \\
1&1 & $\theta_1$, $\theta_2$ \\
1 & $p$, $p>1$ & $\theta_1$, $\theta_2$, $\alpha_2$ \\
$r$, $r>1$ & 1 & $\theta_1$, $\alpha_1$, $\theta_2$ \\
$r$, $r>1$ & $p$, $p>1$ & $\theta_1$, $\alpha_1$, $\theta_2$, $\alpha_2$ 
\end{tabular}
\caption{Identifiable parameters corresponding to different maximal space-time lags $(r,p)$ included in the pairwise likelihood function.}
\label{identpar}
\end{table}
\end{remark}

\begin{proof}[Proof of Theorem \ref{Consistencythm}]
To show strong consistency of the estimates we follow the method of Wald \cite{Wald}. Accordingly, it suffices to show the following conditions. 
\begin{enumerate}[({C}1)]
\item{Strong law of large numbers: Uniformly on the compact set $\Psi$,  
\begin{equation*}
\frac{1}{m^dT}PL^{(m^d,T)}(\bs{\psi}) \stackrel{a.s.}{\longrightarrow} PL(\bs{\psi})\coloneqq \mathbb{E}\left[g_{\bs{\psi}}(1,\ldots,1,1;\mathcal{H}_r,p)\right] , \ \  mT \to \infty.
\end{equation*}}
\item{The function $PL(\bs{\psi})$ is uniquely maximized at the true parameter vector $\bs{\psi}^* \in \bs{\Psi}$.}
\end{enumerate}
From (C1) and (C2) strong consistency follows.
First we prove (C1). 
Recall from \eqref{PLfull2} that the pairwise likelihood function is given by 
\begin{equation*}
PL^{(m^d,T)}(\bs{\psi}) 
= \sum\limits_{i_1=1}^{m}\cdots\sum\limits_{i_d=1}^{m}\sum\limits_{k=1}^{T} g_{\bs{\psi}}\left(i_1,\ldots,i_d,k;\mathcal{H}_r,p\right) - \mathcal{R}^{(m^d,T)}(\bs{\psi}), 
\end{equation*}
where $g_{\bs{\psi}}$ and $\mathcal{R}^{(m^d,T)}(\bs{\psi})$ are defined in \eqref{gpsi} and \eqref{boundary}, respectively.
The pointwise convergence of the first term on the right hand side to $PL(\bs{\psi})$ follows immediately from Proposition \ref{mixing1} together with the fact that $g_{\bs{\psi}}$ in \eqref{gpsi} is a measurable function of lagged versions of  $\eta(\bs{s}_{(i_1,\ldots,i_d),t_k})$.

In the following, we show that the convergence is uniform and that the boundary term defined in \eqref{boundary} converges to zero almost surely.
For both steps, observe first that we can bound the log-density from Lemma \ref{bivdenlemma}. For $x_1,x_2> 0$
\begin{align*}
\left|\log f_{\bs{\psi}}(x_1,x_2)\right| ={} & \left|-V + \log\left(\frac{\partial V}{\partial x_1}\frac{\partial V}{\partial x_2} - \frac{\partial^2 V}{\partial x_1\partial x_2}\right)\right|\\
 \leq{} & \left|-\frac{1}{x_1}\Phi(q_{\psi}^{(1)})\right| + \left|-\frac{1}{x_1}\Phi(q_{\psi}^{(2)})\right| + \left|\frac{\partial V}{\partial x_1}\frac{\partial V}{\partial x_2} - \frac{\partial^2 V}{\partial x_1\partial x_2}\right| \\
\leq{} & \frac{1}{x_1} + \frac{1}{x_2} + \frac{1}{x_1^2x_2^2} + \frac{1}{2\sqrt{\delta(\bs{h},u)}}\left(\frac{1}{x_1^2x_2^2} + \frac{1}{x_1^3x_2} + \frac{1}{x_1^2x_2^2}+\frac{1}{x_1x_2^3} + \frac{1}{x_1^2x_2}+ \frac{1}{x_1 x_2^2} \right)\\
&  + \frac{1}{4\delta(\bs{h},u)}\left(\frac{1}{x_1^2x_2^2} + \frac{1}{x_1^3 x_2} + \frac{1}{x_1x_2^3} 
+ \frac{1}{x_1^2x_2^2} + \left|\frac{q_{\psi}^{(1)}}{x_1^2x_2} + \frac{q_{\psi}^{(2)}}{x_1x_2^2}\right|\right),
\end{align*}
where $q_{\psi}^{(1)}$, $q_{\psi}^{(2)}$ and $V$ are defined in \eqref{q1q2} and $\eqref{V}$, respectively, where $\Phi(\cdot)\leq 1$ and $\varphi(\cdot) \leq 1$ were used.
Since the marginal distributions of the max-stable space-time process are assumed to be standard Fr\'{e}chet, it follows that for every fixed location $\bs{s}\in S$ and fixed time point $t\in T$ $1/\eta(\bs{s},t)$ is standard exponentially distributed. Using H\"older's inequality, it follows that 
\begin{align*}
\mathbb{E}_{\bs{\psi}^*}\left[\left|\log f_{\bs{\psi}}(\eta(\bs{s}_1,t_1),\eta(\bs{s}_2,t_2))\right|\right] \leq{}& K_1 + \frac{K_2}{2\sqrt{\delta(\bs{h},u)}} + \frac{K_3}{4\delta(\bs{h},u)}, 
\end{align*}
where $K_1,K_2,K_3>0$ are finite constants. 
Since the parameter space $\Psi$ is assumed to be compact and together with assumption \eqref{parspace}, $\delta$ can be bounded away from zero, i.e. 
\begin{align}
\delta(\bs{h},u) &= \theta_1\|\bs{h}\|^{\alpha_1} + \theta_2|u|^{\alpha_2} \geq \min\left\{\theta_1,\theta_2\right\}(\|\bs{h}\|^{\alpha_1} + |u|^{\alpha_2}) \nonumber \\
&> c(\|\bs{h}\|^{\alpha_1} + |u|^{\alpha_2})  > \tilde{c} >0, 
\label{boundarydelta}
\end{align} 
since $\alpha_1,\alpha_2 \in (0,2]$, where $\tilde{c}>0$ is some constant independent of the parameters. 
Therefore, 
\begin{align}
\mathbb{E}_{\bs{\psi}^*}\left[\left|\log f_{\bs{\psi}}(\eta(\bs{s}_1,t_1),\eta(\bs{s}_2,t_2))\right|\right]<{}& K_1 + \frac{K_2}{2\sqrt{\tilde{c}}} + \frac{K_3}{4\tilde{c}} \eqqcolon K_4 < \infty, 
\label{expfinite}
\end{align}
where $K_4>0$. 
Note that in the same way we can show that the expectation of the squared bivariate log-density is finite, since it only involves higher order moments of the exponential distribution. 

To establish uniform convergence, we follow Straumann and Mikosch \cite{Straumann1}, Theorem 2.7, and show that 
$$\mathbb{E}\left[\sup\limits_{\bs{\psi} \in \bs{\Psi}}\left|g_{\bs{\psi}}(1,\ldots,1,1;\mathcal{H}_r,p)\right|\right] < \infty.$$
It is sufficient to verify that
$$\mathbb{E}\left[\sup\limits_{\bs{\psi} \in \bs{\Psi}}\left|\log f_{\bs{\psi}}(\eta(\bs{s}_{(i_1,\ldots,i_d)},t_k),\eta(\bs{s}_{(i_1,\ldots,i_d)}+\bs{h},t_k+u))\right|\right] < \infty.$$
Since the pairwise density is continous and because of the compact parameter space, the statement follows immediately using \eqref{boundarydelta} and \eqref{expfinite}.  

As a last step for (C1) we show that the boundary term $\mathcal{R}^{(m^d,T)}(\bs{\psi})$ converges to 0 almost surely. For notational simplicity, we only consider the case $d=2$. The general case is proved analogously. First note from \eqref{boundary} that 
\begin{align*}
& \mathbb{E}\left[\left|\frac{1}{m^2T}\mathcal{R}^{(m^2,T)}(\bs{\psi})\right| \right]\\ {}&\leq\frac{1}{m^2T}\sum\limits_{i_1=1}^{m}\sum\limits_{i_2=1}^m\sum\limits_{\stackrel{\bs{h}\in\mathcal{H}_r}{\bs{s}_{(i_1,i_2)}+\bs{h}\notin S}}\sum\limits_{k=1}^{T}\sum\limits_{\stackrel{l=k+1}{l>T}}^{k+p}\mathbb{E}\left[\left|\log f_{\bs{\psi}}(\eta(\bs{s}_{(i_1,i_2)},t_k),\eta(\bs{s}_{(i_1,i_2)}+\bs{h},t_l))\right|\right] \\
{}& \leq \frac{1}{m^2T} \sum\limits_{i_1=1}^{m}\sum\limits_{i_2=1}^m\sum\limits_{\stackrel{\bs{h}\in\mathcal{H}_r}{\bs{s}_{(i_1,i_2)}+\bs{h}\notin S}}\sum\limits_{k=1}^{T}\sum\limits_{\stackrel{l=k+1}{l>T}}^{k+p}K_4  \leq \frac{K_4K_5}{mT} \to 0, \ \ mT  \to \infty, 
\end{align*} 
where we used the bound derived in \eqref{expfinite} and the fact that the number of space-time points in the boundary is of order $m$ (independent of $T$) and, therefore, can be bounded by $K_5m$ with $K_5>0$ a constant independent of $m$ and $T$. 
We write $\mathcal{R}^{(m^2,T)}(\bs{\psi})$ in the following way using the function $g_{\bs{\psi}}$ in \eqref{gpsi}. Denote by $\mathcal{B}_{m,T}$ the set of space-time indices $(i_1,i_2,k)$ for which $\bs{s}_{i_1,i_2} + \bs{h}\notin S$ or $k>T$. In Figure \ref{gridskizze} $\mathcal{B}_{m,T}$ corresponds to the indices of the locations $\bs{s}_{(i_1,i_2)}\in S$ for which $\bs{s}_{(i_1,i_2)}+\bs{h}$ is in the outer polygon. The cardinality of the set $\mathcal{B}_{m,T}$ can be bounded by using the maximum norm instead of the euclidean norm, i.e.
$$|\mathcal{B}_{m,T}| \leq r(2m + 1) \eqqcolon K_5m. $$
The boundary term in \eqref{boundary} can then be written as 
$$\mathcal{R}^{(m^2,T)}(\bs{\psi}) = \sum\limits_{(i_1,i_2,k)\in \mathcal{B}_{m,T}} g_{\bs{\psi}}(i_1,i_2,k;\mathcal{H}_r,p),$$
where $g_{\bs{\psi}}$ is defined in \eqref{gpsi}. 
In the same way as before, it follows by the strong law of large numbers that
$$\frac{1}{|\mathcal{B}_{m,T}|}\sum\limits_{(i_1,i_2,k)\in\mathcal{B}_{m,T}} g_{\bs{\psi}}(i_1,i_2,k;\mathcal{H}_r,p) \stackrel{a.s.}{\longrightarrow} \mathbb{E}\left[g_{\bs{\psi}}(1,1,1;\mathcal{H}_r,p)\right],$$
uniformly on the compact set $\Psi$. 
Therefore, 
\begin{align*}
\frac{1}{m^2T}\mathcal{R}^{(m^2,T)}(\bs{\psi}) &\leq  \frac{K_5 }{m T}\frac{1}{|\mathcal{B}_{m,T}|}\sum\limits_{(i_1,i_2,k)\in \mathcal{B}_{m,T}} g_{\bs{\psi}}(i_1,i_2,k;\mathcal{H}_r,p) \stackrel{a.s.}{\longrightarrow} 0,
\end{align*}
since $\mathbb{E}\left[|g_{\bs{\psi}}(1,1,1;\mathcal{H}_r,p)|\right]<\infty$.
This proves (C1). 

To prove (C2), note by Jensen's inequality that 
\begin{align*}
& \mathbb{E}_{\bs{\psi}^*}\left[\log\left(\frac{f_{\bs{\psi}}(x_1,x_2)}{f_{\bs{\psi^*}}(x_1,x_2)}\right)\right] \leq 
\log\left(\mathbb{E}_{\bs{\psi}^*}\left[\frac{f_{\bs{\psi}}(x_1,x_2)}{f_{\bs{\psi^*}}(x_1,x_2)}\right]\right) = 0 \end{align*} 
and, hence, $$PL(\bs{\psi}) \leq PL(\bs{\psi}^*).$$
So, $\bs{\psi}^*$ maximizes $PL(\bs{\psi})$ and is the unique optimum if and only if there is equality in Jensen's inequality. However, this is precluded by \eqref{ident1}.
\end{proof}

\section{Asymptotic normality of the pairwise likelihood estimates}\label{Asymptoticnormality}

In order to prove asymptotic normality of the pairwise likelihood estimates resulting from maximizing \eqref{PLfull2} we need the following results for the pairwise log-density. The proofs can be found in Appendix \ref{appendix}.
\begin{lemma}\label{finiteder}
\begin{enumerate}[(1)]
\item{The gradient of the bivariate log-density satisfies 
$$\mathbb{E}_{\bs{\psi}^*}\left[\left|\nabla_{\bs{\psi}} \log f_{\bs{\psi}}(\eta(\bs{s}_1,t_1),\eta(\bs{s}_2,t_2))\right|^{3}\right]<\infty$$}
\item{The Hessian of the pairwise log-density satisfies $$\mathbb{E}_{\bs{\psi}^*}\left[\sup\limits_{\bs{\psi}\in\bs{\Psi}}\left|\nabla^2_{\bs{\psi}} \log f_{\bs{\psi}}(\eta(\bs{s}_1,t_1),\eta(\bs{s}_2,t_2))\right|\right]< \infty.$$}
\end{enumerate}
\end{lemma}
Assuming asymptotic normality of the pairwise score function, then it is relatively routine to show that the pairwise likelihood estimates are asymptotically normal. To formulate the result, recall from \eqref{PLfull2} that the pairwise likelihood function can be written as 
$$PL^{(m^d,T)}(\bs{\psi}) = \sum\limits_{i_1=1}^m\cdots \sum\limits_{i_d=1}^m \sum\limits_{k=1}^T g_{\bs{\psi}}(i_1,\ldots,i_d,k;\mathcal{H}_r,p) - \mathcal{R}^{(m^d,T)}(\bs{\psi}),$$
where $g_{\bs{\psi}}$ is defined in \eqref{gpsi}. The pairwise score function is then given by 
$$\sum\limits_{i_1=1}^m\cdots \sum\limits_{i_d=1}^m \sum\limits_{k=1}^T\nabla_{\bs{\psi}} g_{\bs{\psi}}(i_1,\ldots,i_d,k;\mathcal{H}_r,p) -\nabla_{\bs{\psi}}\mathcal{R}^{(m^d,T)}(\bs{\psi}),$$
where $\nabla_{\bs{\psi}} g_{\bs{\psi}}(i_1,\ldots,i_d,k;\mathcal{H}_r,p)$ is the gradient of the function $g_{\bs{\psi}}$ with respect to $\bs{\psi}$. 
\begin{theorem}\label{normality1}
Assume that the conditions of Theorem \ref{Consistencythm} hold. In addition, assume that a central limit theorem holds for 
$\nabla_{\bs{\psi}} g_{\bs{\psi}}(i_1,\ldots,i_d,k;\mathcal{H}_r,p)$ in the following sense 
\begin{equation}
\frac{1}{m^{d/2}\sqrt{T}} \sum\limits_{i_1=1}^m\cdots\sum\limits_{i_d=1}^m \sum\limits_{k=1}^T \nabla_{\bs{\psi}}g_{\bs{\psi}^*}(i_1,\ldots,i_d,k;\mathcal{H}_r,p)  \stackrel{d}{\longrightarrow}\mathcal{N}(0,\Sigma), \ mT\to \infty,
\label{CLT1}
\end{equation}
where $\bs{\psi}^*$ is the true parameter vector and $\Sigma$ is some covariance matrix. Then it follows that the pairwise likelihood estimates $\hat{\bs{\psi}}_{m^d,T}$ satisfy
$$m^{d/2}\sqrt{T}(\bs{\hat{\psi}}_{m^d,T}-\bs{\psi}^*) \stackrel{d}{\longrightarrow} \mathcal{N}(0,F^{-1}\Sigma(F^{-1})^T), \ mT\to \infty,$$
where 
$$F = \mathbb{E}_{\bs{\psi}^*}\left[-\nabla^2_{\bs{\psi}}g_{\bs{{\psi}}^*}(1,\ldots,1,1;\mathcal{H}_r,p)\right].$$
\end{theorem}

\begin{proof}
We use a standard Taylor expansion of the pairwise score function around the true parameter vector:
\begin{align*}
m^{d/2}\sqrt{T}(\bs{\hat{\psi}}_{m^d,T}- \bs{\psi}^*) =& -\left(\frac{1}{m^dT}\nabla_{\bs{\psi}}^2PL^{(m^d,T )}(\bs{\tilde{\psi}})\right)^{-1}\left(\frac{1}{m^{d/2}\sqrt{T}}\nabla_{\bs{\psi}}PL^{(m^d,T)}(\bs{\psi}^*)\right)\\
=& -\left(\frac{1}{m^dT}\sum\limits_{i_1=1}^m\cdots\sum\limits_{i_d=1}^m\sum\limits_{k=1}^T  \nabla_{\bs{\psi}}^2g_{\bs{\tilde{\psi}}}(i_1,\ldots,i_d,k;\mathcal{H}_r,p)-\nabla^2_{\bs{\psi}}\mathcal{R}^{(m^d,T)}(\bs{\tilde{\psi}})\right)^{-1} \\
&\times \left(\frac{1}{m^{d/2}\sqrt{T}}\sum\limits_{i_1=1}^m\cdots\sum\limits_{i_2=1}^m\sum\limits_{k=1}^T \nabla_{\bs{\psi}}g_{\bs{\psi}^*}(i_1,\ldots,i_d,k;\mathcal{H}_r,p)-\nabla_{\bs{\psi}}\mathcal{R}^{(m^d,T)}(\bs{\psi}^*)\right),
\end{align*}
where $\tilde{\bs{\psi}} \in [\bs{\hat{\psi}}_{m^d,T},\bs{\psi}^*]$. 
For now, we ignore the boundary term and analyze the first terms in the outer brackets.
By \eqref{CLT1} the second term converges to a normal distribution with mean $0$ and covariance matrix $\Sigma$. 
For the first part we use the same arguments as in the consistency proof and show a strong law of large numbers. Since the underlying space-time process in the likelihood function is mixing, it follows that the process \\ $\left\{\nabla^2_{\bs{\psi}} g_{\bs{\psi}}(i_1,\ldots,i_d,k;\mathcal{H}_r,p), \bs{s}_{(i_1,\ldots,i_d)}\in \bbz^d,t_k \in \bbz\right\}$ is mixing as a measurable function of mixing and lagged processes. To prove the uniform convergence we need to verify that
$$\mathbb{E}_{\bs{\psi}^*}\left[\sup\limits_{\bs{\psi}\in \bs{\Psi}}\left|\nabla_{\bs{\psi}}^2 g_{\bs{\psi}}(1,\ldots,1,1;\mathcal{H}_r,p)\right|\right] < \infty.$$
This follows immediately from Lemma \ref{finiteder}. 
Putting this together with the fact that $\tilde{\bs{\psi}} \in [\bs{\hat{\psi}}_{m^d,T},\bs{\psi}^*]$, and because of the strong consistency, it follows that
\begin{align*}
\frac{1}{m^dT}\sum\limits_{i_1=1}^m\cdots\sum\limits_{i_d=1}^m\sum\limits_{k=1}^T  \nabla_{\bs{\psi}}^2g_{\bs{\tilde{\psi}}}(i_1,\ldots,i_d,k;\mathcal{H}_r,p) \stackrel{a.s.}{\longrightarrow} \mathbb{E}_{\bs{\psi}*}\left[\nabla^2_{\bs{\psi}}g_{\bs{\psi}^*}(1,\ldots,1,1;\mathcal{H}_r,p)\right] \eqqcolon -F.
\end{align*}
For the boundary term $\mathcal{R}^{(m^d,T)}$ observe that it can be written as
$$\mathcal{R}^{(m^d,T)}(\bs{\psi}) = \sum\limits_{(i_1,\ldots,i_d,k)\in \mathcal{B}_{m,T}}g_{\bs{\psi}}(i_1,\ldots,i_d,k;\mathcal{H}_r,p).$$
Using assumption \eqref{CLT1} together with the strong law of large numbers for $\left\{\nabla^2_{\bs{\psi}}g_{\bs{\psi}}(i_1,\ldots,i_d,k;\mathcal{H}_r,p)\right\}$ it follows in the same way as in the proof of Theorem \ref{Consistencythm} that
\begin{align*}
&\frac{1}{m^dT}\nabla^2_{\bs{\psi}} \mathcal{R}^{(m^d,T)} \stackrel{a.s}{\longrightarrow} 0, \ \ \text{and} \ \
\frac{1}{m^{d/2}\sqrt{T}} \nabla_{\bs{\psi}} \mathcal{R}^{(m^d,T)} \stackrel{P}{\longrightarrow} 0. 
\end{align*}
Combining these results, we obtain
$$m^{d/2}\sqrt{T}(\bs{\hat{\psi}}_{m^d,T}-\bs{\psi}^*) \stackrel{d}{\longrightarrow} \mathcal{N}(0,F^{-1}\Sigma (F^{-1})^T), \ \ mT\to \infty.$$
\end{proof}

In the next section we provide a sufficient condition for \eqref{CLT1}.

\subsection{Asymptotic normality and $\alpha$-mixing}
In this section we consider asymptotic normality of the parameters estimates for the process in  Proposition \ref{BrownResnick}. Under the assumption of $\alpha$-mixing on the random field the key is to show asymptotic normality for the score function of the pairwise likelihood. For an increasing time domain and fixed number of locations asymptotic normality of the pairwise likelihood estimates was shown in Huser and Davison \cite{Huser}. 
The main difference between a temporal setting and a space-time setting is the definition of the $\alpha$-mixing coefficients and the resulting assumptions to obtain a central limit theorem for the score function. 

We apply the central limit theorem for random fields established in Bolthausen \cite{Bolthausen} to the score function of the pairwise likelihood in our model. In a second step we verify the $\alpha$-mixing conditions for the max-stable process introduced in Section \ref{Desmodel}.  
First, we define the $\alpha$-mixing coefficients in a space-time setting as follows. 
Define the distances 
\begin{align*}
d((\bs{s}_1,t_1),(\bs{s}_2,t_2)) &= \max\left\{\max\limits_{1\leq i\leq d} |\bs{s}_1(i)-\bs{s}_2(i)|,|t_1-t_2|\right\} , \quad \bs{s}_1,\bs{s}_2 \in \bbz^d, t_1,t_2 \in \bbn\\
d(\Lambda_1,\Lambda_2) &= \inf\left\{d((\bs{s}_1,t_1),(\bs{s}_2,t_2)), (\bs{s}_1,t_1) \in \Lambda_1, (\bs{s}_2,t_2) \in \Lambda_2\right\}, \quad \Lambda_1,\Lambda_2 \subset \bbz^{d}\times \bbn.
\end{align*}
Let further $\mathcal{F}_{\Lambda_i}= \sigma\left\{\eta(\bs{s},t), (\bs{s},t)\in \Lambda_i\right\}$ for $i=1,2$. 
The mixing coefficients are defined for $k,l,n\geq 0$ by 
\begin{equation}\label{alphaBolt}
\alpha_{k,l}(n) = \sup\left\{\left|P(A_1\cap A_2) - P(A_1)P(A_2)\right|: \ A_i \in \mathcal{F}_{\Lambda_i}, |\Lambda_1|\leq k, |\Lambda_2|\leq l, d(\Lambda_1,\Lambda_2) \geq n\right\}
\end{equation}
and depend on the sizes and the distance of the sets $\Lambda_1$ and $\Lambda_2$.  
A space-time process is called $\alpha$-mixing, if $\alpha_{k,l}(n) \to 0$ as $n\to\infty$ for all $k,l\geq 0$.
We assume that the process $\left\{\eta(\bs{s},t), (\bs{s},t)\in \bbz^{d}\times \bbn\right\}$ is $\alpha$-mixing with mixing coefficients defined in \eqref{alphaBolt}, from which it follows that the score process 
\begin{equation}
\left\{\nabla_{\bs{\psi}}g_{\psi}(i_1,\ldots,i_d,k;\mathcal{H}_r,p), (\bs{s}_{(i_1,\ldots,i_d)},t_k) \in \bbz^{d}\times \bbn\right\}.
\label{spaceprocess}
\end{equation}
is $\alpha$-mixing.
We apply Bolthausen's central limit theorem to the process in \eqref{spaceprocess}. By adjusting the assumptions on the $\alpha$-mixing coefficients to the score process, we obtain the following proposition. 

\begin{proposition}\label{Bolthausen2}
Assume, that the following conditions hold: 
\begin{enumerate}[(1)]
\item{The process $\left\{(\eta(\bs{s},t), (\bs{s},t)\in\bbz^{d}\times \bbn\right\}$ is strongly mixing with mixing coefficients $\alpha_{k,l}(n)$ as in \eqref{alphaBolt}.} 
\item{$\sum\limits_{n=1}^{\infty}n^{d}\alpha_{k,l}(n) <\infty \text{ for }\ k+l\leq 4(|\mathcal{H}_r|+1)(p+1)$ and   $\alpha_{(|\mathcal{H}_r|+1)(p+1),\infty}(n) = o(n^{-(d+1)})$.}
\item{There exists $\beta>0$ such that 
\begin{align*}
&\mathbb{E}\left[\left|\nabla_{\bs{\psi}}g_{\bs{\psi}^*}((i_1,\ldots,i_d),k,\mathcal{H}_r,p)\right|^{2+\beta}\right] <\infty  \ \ \text{ and }\\
& \sum\limits_{n=1}^{\infty}n^{d}\alpha_{(|\mathcal{H}_r|+1)(p+1),(|\mathcal{H}_r|+1)(p+1)}(n)^{\beta/(2+\beta)} < \infty.
\end{align*}}
\end{enumerate}
Then,
$$\frac{1}{ m^{d/2}\sqrt{T}} \sum\limits_{i_1=1}^m\cdots\sum\limits_{i_d=1}^m \sum\limits_{k=1}^T \nabla_{\bs{\psi}}g_{\bs{\psi}^*}(i_1,\ldots,i_d,k;\mathcal{H}_r,0) \stackrel{d}{\to} \mathcal{N}(0,\Sigma), \ mT\to \infty, $$
where $\Sigma = \sum\limits_{\bs{s}_{(i_1,\ldots,i_d)}\in \bbz^{d}}\sum\limits_{t_k\in\bbn}\mathbb{C}ov\left( \nabla_{\bs{\psi}}g_{\bs{\psi}^*}(1,\ldots,1,1;\mathcal{H}_r,p), \nabla_{\bs{\psi}}g_{\bs{\psi}^*}(i_1,\ldots,i_d,k;\mathcal{H}_r,p)\right)$.
\end{proposition}

In a second step, we want to analyze the strong mixing property and the related assumptions for our model. 
Recent work by Dombry and Eyi-Minko \cite{Dombry} deals with strong mixing properties for max-stable random fields. By using point process representation of max-stable processes together with coupling techniques, they showed that the $\alpha$-mixing coefficients can be bounded by a function of the tail dependence coefficient. A direct extension to the space-time setting gives the following lemma. 

\begin{lemma}[Dombry and Eyi-Minko \cite{Dombry}, Corollary 2.2]
Consider the stationary max-stable space-time process $\left\{\eta(\bs{s},t), (\bs{s},t)\in \bbz^d\times \bbn\right\}$ with tail dependence coefficient $\chi(\bs{h},u)$. 
The $\alpha$-mixing coefficients in \eqref{alphaBolt} satisfy
$$\alpha_{k,l}(n) \leq kl\sup\limits_{\max\left\{\|\bs{h}\|,|u|\right\} \geq n}\chi(\bs{h},u) \ \  \text{ and }\ \ \alpha_{k,\infty}(n) \leq k\sum\limits_{\max\left\{\|\bs{h}\|,|u|\right\}\geq n}\chi(\bs{h},u).$$
\end{lemma}

For the model described in Proposition \ref{BrownResnick} with tail dependence coefficient $\chi$ in \eqref{chi}, it follows by using the inequality for the normal tail probability $\overline{\Phi}(x)\leq e^{-x^2/2}$ that
\begin{align*}
\alpha_{k,l}(n) & \leq 4kl\sup\limits_{\max\left\{\|\bs{h}\|,|u|\right\} \geq n} (1-\Phi(\sqrt{\delta(\bs{h},u)})) \leq 4kl\sup\limits_{\max\left\{\|\bs{h}\|,|u|\right\} \geq n}\exp\left\{-\frac{\delta(\bs{h},u)}{2}\right\}\\
& = 4kl\sup\limits_{\max\left\{\|\bs{h}\|,|u|\right\} \geq n}\exp\left\{-\frac{1}{2}(\theta_1\|\bs{h}\|^{\alpha_1}+\theta_2|u|^{\alpha_2})\right\}\\
&\leq 4kl \sup\limits_{\max\left\{\|\bs{h}\|,|u|\right\}\geq n} \exp\left\{-\frac{1}{2}\min\left\{\theta_1,\theta_2\right\}(\max\left\{\|\bs{h}\|,|u|\right\})^{\min\left\{\alpha_1,\alpha_2\right\}}\right\}.  
\end{align*}
For $n\to \infty$, this tends to zero for all $k,l\geq 0$. Thus, $\left\{\eta(\bs{s},t),(\bs{s},t)\in \bbz^{d}\times \bbn\right\}$ is strongly mixing. 
This shows the first assertion in Corollary \ref{Bolthausen2}. 
Furthermore, for $k+l\leq 4(|\mathcal{H}_r|+1)(p+1)$ the coefficients satisfy 
\begin{align*}
\sum\limits_{n=1}^{\infty}n^{d}\alpha_{k,l}(n) & \leq 4kl\sum\limits_{n=1}^{\infty} n^{d}\sup\limits_{\max\left\{\|\bs{h}\|,|u|\right\}\geq n}\exp\left\{-\frac{1}{2}(\theta_1\|\bs{h}\|^{\alpha_1}+\theta_2|u|^{\alpha_2})\right\}\\
& \leq 4kl\sum\limits_{n=1}^{\infty}n^{d}\exp\left\{-\frac{1}{2}\min\left\{\theta_1,\theta_2\right\}n^{\min\left\{\alpha_1,\alpha_2\right\}}\right\}<\infty.
\end{align*}
In addition, 
\begin{align*}
n^{d+1}\alpha_{(|\mathcal{H}_r|+1)(p+1),\infty}(n)
& \leq n^{d+1} (|\mathcal{H}_r|+1)(p+1)\sum\limits_{x\geq n}\exp\left\{-\frac{1}{2}\min\left\{\theta_1,\theta_2\right\}x^{\min\left\{\alpha_1,\alpha_2\right\}}\right\},
\end{align*}
which proves (2) in Proposition \ref{Bolthausen2}. 
As for (3), from Lemma \ref{finiteder} and using $\beta=1$ we know that
$$\mathbb{E}\left[\left|\nabla_{\bs{\psi}}g_{\bs{\psi}^*}((i_1,\ldots,i_d),k,\mathcal{H}_r,p)\right|^{(2+\beta)}\right] <\infty.$$
Using the same arguments as above, it is easy to see that the second condition in (3) holds.  
By combining the above results with Theorem \ref{normality1} we obtain asymptotic normality for the parameter estimates $\bs{\hat{\psi}}_{m^d,T}$ for an increasing number of space-time locations. 
\begin{theorem}
Assume that the conditions of Theorem \ref{Consistencythm} hold. Then, 
$$(m^d T)^{1/2}(\hat{\bs{\psi}}_{m^d,T}-\bs{\psi}^*)\stackrel{d}{\to} \mathcal{N}(0,F^{-1}\Sigma (F^{-1})^\top), \ \ mT\to \infty,$$
with 
$$F = \mathbb{E}_{\bs{\psi}^*}\left[-\nabla^2_{\bs{\psi}}g_{\bs{{\psi}}^*}(1,\ldots,1,1;\mathcal{H}_r,p)\right]$$
and 
$$\Sigma = \sum\limits_{\bs{s}_{(i_1,\ldots,i_d)}\in \bbz^d}\sum\limits_{t_k\in\bbn}
\mathbb{C}ov\left( \nabla_{\bs{\psi}}g_{\bs{\psi}^*}(1,\ldots,1,1;\mathcal{H}_r,p), \nabla_{\bs{\psi}}g_{\bs{\psi}^*}(i_1,\ldots,i_d,k;\mathcal{H}_r,p)\right).$$
\end{theorem}

\section{Simulation study}\label{Simulation}
We illustrate the small sample behaviour of the pairwise likelihood estimation for spatial dimension $d=2$ in a simulation experiment. 
The setup for this study is: 
\begin{enumerate}
\item{The spatial locations consisted of a $10\times 10$ grid 
$$S = \left\{\bs{s}_{(i_1,i_2)}=(i_1,i_2), i_1,i_2 \in \left\{1,\ldots,10\right\}\right\}.$$ The time points are chosen equidistantly, $1<\cdots T=100$. 
}
\item{One hundred independent Gaussian space-time processes $Z_j(s_n\bs{s},t_nt), j=1,\ldots,100$ were generated using the R-package \verb RandomFields ~with covariance function $\rho(s_n\bs{h},t_nu)$. 
We use the following correlation function for the underlying Gaussian random field. 
$$\rho(\bs{h},u) = (1+\theta_1\|\bs{h}\|^{\alpha_1} + \theta_2|u|^{\alpha_2})^{-3/2}.$$
Assumption \ref{ass1} is fullfilled and the limit function $\delta$ is given by 
$$\lim\limits_{n\to\infty} \log n(1-\rho(s_n\bs{h},t_nu)) = \delta(\bs{h},u) = \frac{3}{2}\theta_1\|\bs{h}\|^{\alpha_1} + \frac{3}{2}\theta_2|u|^{\alpha_2}. $$}
\item{The simulated processes were transformed to standard Fr\'echet margins using the transformation 
$-1/\log(\Phi(Z_j(\bs{s},t)))$ for $\bs{s}\in S$ and $t\in \left\{t_1,\ldots,t_T\right\}$.}
\item{The pointwise maximum of the transformed Gaussian random fields was computed and rescaled by $1/n$ to obtain an approximation of a max-stable random field, i.e.
$$\eta(\bs{s},t) = \frac{1}{100}\bigvee\limits_{j=1}^{100} -\frac{1}{\log\left(\Phi(Z_j(s_n\bs{s},t_nt))\right)}, \ \bs{s}\in S, t\in \left\{t_1,\ldots,t_T\right\}.$$}
\item{The parameters $\theta_1,\alpha_1,\theta_2$ and $\alpha_2$ for different combinations of maximal space-time lags $(r,p)$ were estimated by maximizing \eqref{PLfull2}. The program is adjusted such that it takes care of identifiability issues, when some of the parameters are not identifiable, cf. Remark \ref{Identifiability}.}
\item{Steps 1. - 5. are repeated 100 times.}
\end{enumerate}

Figures \ref{Sim1} and \ref{Sim2} show the resulting estimates as a function of $(r,p)$, where the true parameter set is given by $\bs{\psi}^* = (\theta_1^*,\alpha_1^*,\theta_2^*,\alpha_2^*) =(0.06,1,0.04,1)$.
Figure \ref{Sim1} shows the resulting estimates for the spatial parameters $\theta_1$ and $\alpha_1$. The  horizontal axis shows the different maximal space-time lags included in the pairwise likelihood function from \eqref{PLfull2}. Each vertical dot shows the result for one specific simulation. 
The dotted lines show confidence bands based on the simulation results. 
In addition to the graphical output we calculate the root mean square error (RMSE) and the mean abolute error (MAE) to see how the choice of $(r,p)$ influences the estimation. 

We draw the following conclusions. As already pointed out by Davis and Yau \cite{Davis} and Huser and Davison \cite{Huser}, there might be a loss in efficiency if too many pairs are included in the estimation. This can be explained by the fact that pairs get more and more independent as the space-time lag increases. Adding more and more pairs to the pairwise log-likelihood function can introduce some noise which descreases the efficiency.  
This is evident in Figure \ref{Sim2} for the temporal parameter $\alpha_2$, where the estimates vary more around the mean as more pairs are included in the estimation. 

An interesting observation for our model is that using a maximal spatial lag of $\bs{0}$ or a maximal temporal lag $0$, respectively, leads to very good results. For the spatial parameters, the space-time lags which lead to the lowest RMSE and MAE are $(2,0)$ for $\theta_1$ and $(2,0)$ (RMSE) or $(3,0)$ (MAE) for $\alpha_2$ (see Table \ref{RMSE1}), i.e. we use all pairs within a spatial distance of $2$ or $3$ at the same time point. Basically, this suggests that we could also estimate the spatial parameters based on each individual random field for fixed time points and then take the mean over all estimates in time. The same holds for the time parameters $\theta_2$ and $\alpha_2$, where the best results in the sense of the lowest RMSE and MAE are obtained for the space-time lags $(0,3)$, i.e. if we use all pairwise densities corresponding to the space-time pairs $(\bs{s},t_1)$ and $(\bs{s},t_2)$, where $|t_2-t_1|\leq 3$ (see Table \ref{RMSE2}). 
The reason for this observation is that the parameters of the underlying space-time correlation function get ``separated'' in the extremal setting in the sense that for example a spatial lag equal to zero does not affect the temporal parameters $\theta_1$ and $\alpha_1$ and vice versa.

\begin{table}[ht]
\centering
\begin{tabular}{c|c|c|c|c|c|c|c|c|c}
$\theta_1 $ & (1,0) & (1,1) & (1,2) & (1,3) & (1,4) & (1,5) & (2,0) & (2,1) & (2,2)  \\ 
\hline
RMSE &0.0123 &0.0118 &0.0121 &0.0122 &0.0123 &0.0124 &0.0103 &0.0104 &0.0105 \\
MAE & 0.0105 &0.0090 &0.0092 &0.0093 &0.0094 &0.0095 &0.0080 &0.0081 &0.0081 \\
\hline 
& (2,3) & (2,4) & (2,5) & (3,0) & (3,1) & (3,2) & (3,3) & (3,4) & (3,5)  \\
\hline
RMSE &  0.0104 &0.0104 &0.0104 &0.0106 &0.0107 &0.0108 &0.0108 &0.0107 &0.0108 \\
MAE & 0.0081 &0.0081 &0.0081 &0.0082 &0.0083 &0.0083 &0.0084 &0.0083 &0.0084 \\
\hline
\hline
\end{tabular}
\begin{tabular}{c|c|c|c|c|c|c|c|c|c}
$\alpha_1 $ & (2,0) & (2,1) & (2,2) & (2,3) & (2,4) & (2,5) & (3,0) & (3,1) & (3,2)  \\ 
\hline
RMSE & 0.1338 &0.1398 &0.1530 &0.1492 &0.1543 &0.1569 &0.1351 &0.1409 &0.1579  \\
MAE & 0.1078 &0.1124 &0.1154 &0.1137 &0.1233 &0.1252 &0.1050 &0.1106 &0.1127 \\
\hline 
& (3,3) & (3,4) & (3,5) & (4,0) & (4,1) & (4,2) & (4,3) & (4,4) & (4,5)  \\
\hline
RMSE &  0.1596 &0.1639 &0.1649 &0.1423 &0.1483 &0.1614 &0.1673 &0.1735 &0.1751 \\
MAE & 0.1228 &0.1291& 0.1297& 0.1120 &0.1176& 0.1114 &0.1276 &0.1372& 0.1385
\end{tabular}
\caption{RMSE and MAE based on 100 simulations for the spatial estimates $\theta_1$ and $\alpha_1$ for different combinations of maximal space-time lags $(r,p)$. }
\label{RMSE1}
\end{table}

\begin{table}[ht]
\centering
\begin{tabular}{c|c|c|c|c|c|c|c|c|c}
$\hat{\theta}_2 $ & (0,1) & (0,2) & (0,3) & (0,4) & (0,5) & (1,1) & (1,2) & (1,3) & (1,4)  \\ 
\hline
RMSE &0.0182 &0.0182 &0.0182 &0.0182 &0.0182 &0.0184 &0.0183 &0.0183 &0.0183 \\
MAE & 0.0171 &0.0171& 0.0171& 0.0171& 0.0171& 0.0173& 0.0172& 0.0171& 0.0171 \\
\hline 
& (1,5) & (2,1) & (2,2) & (2,3) & (2,4) & (2,5) & (3,1) & (3,2) & (3,3)  \\
\hline
RMSE &  0.0183 &0.0187 &0.0186 &0.0185 &0.0185 &0.0185 &0.0188 &0.0188& 0.0186\\
MAE & 0.0171& 0.0175& 0.0174& 0.0173& 0.0174& 0.0173& 0.0176& 0.0176& 0.0174\\
\hline
\hline 
\end{tabular}
\begin{tabular}{c|c|c|c|c|c|c|c|c|c}
$\hat{\alpha}_2 $ & (0,2) & (0,3) & (0,4) & (0,5) & (1,2) & (1,3) & (1,4) & (1,5) & (2,2)  \\ 
\hline
RMSE & 0.1317 &0.1269 &0.1280 &0.1289 &0.1442 &0.1401& 0.1426& 0.1438& 0.1463  \\
MAE & 0.1008& 0.0989 &0.1015 &0.1035 &0.1086& 0.1079 &0.1139 &0.1147& 0.1179 \\
\hline 
& (2,3) & (2,4) & (2,5) & (3,2) & (3,3) & (3,4) & (3,5) & (4,2) & (4,3)  \\
\hline
RMSE & 0.1532& 0.1580& 0.1619 &0.1473&0.1531& 0.1589 &0.1642 &0.1549 &0.1607 \\
MAE & 0.1242& 0.1275& 0.1294& 0.1169& 0.1223 &0.1273& 0.1317& 0.1233& 0.1284
\end{tabular}
\caption{RMSE and MAE based on 100 simulations for the spatial estimates $\theta_2$ and $\alpha_2$ for different combinations of maximal space-time lags $(r,p)$. }
\label{RMSE2}
\end{table}

\begin{figure}[htp]
\centering
\includegraphics[width=14.5cm, totalheight=8cm]{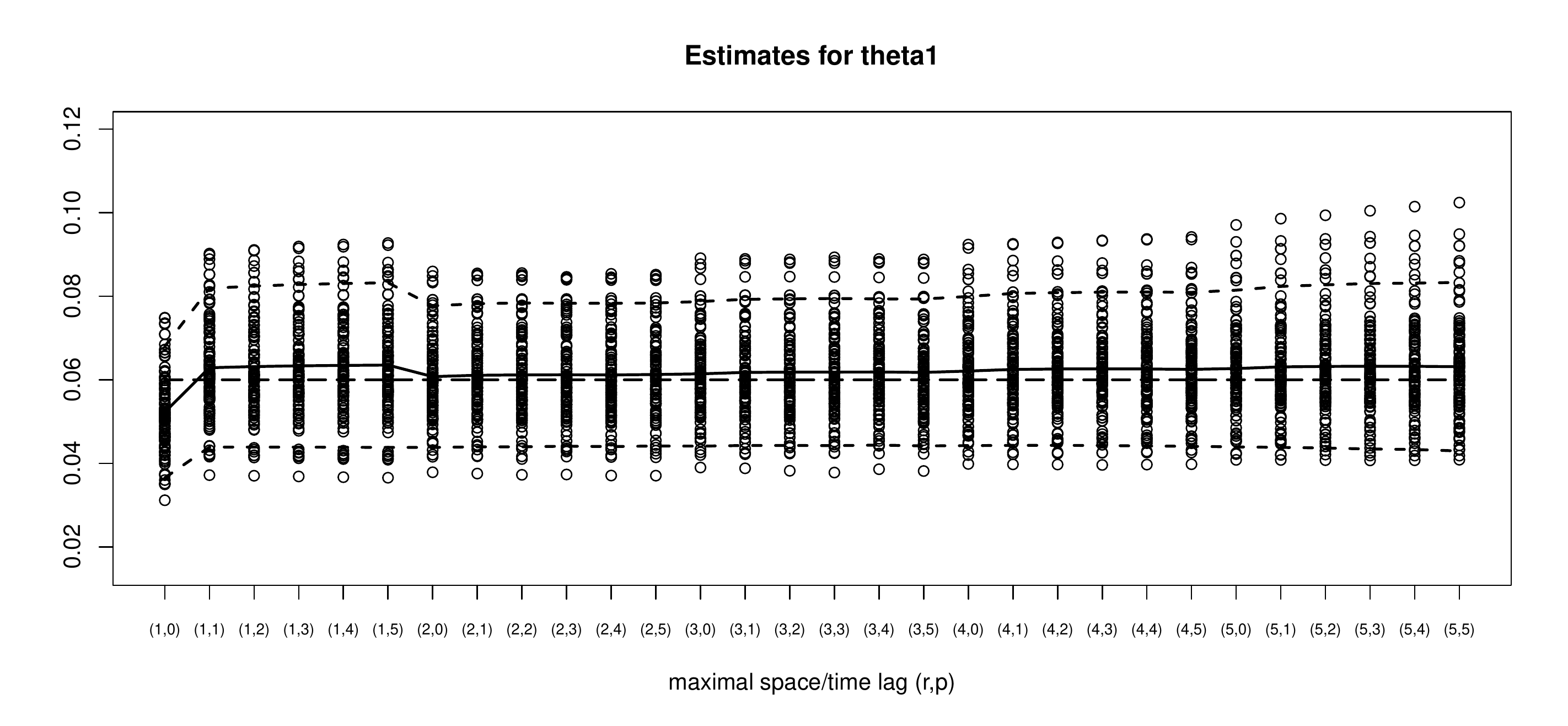}
\includegraphics[width=14.5cm, totalheight=8cm]{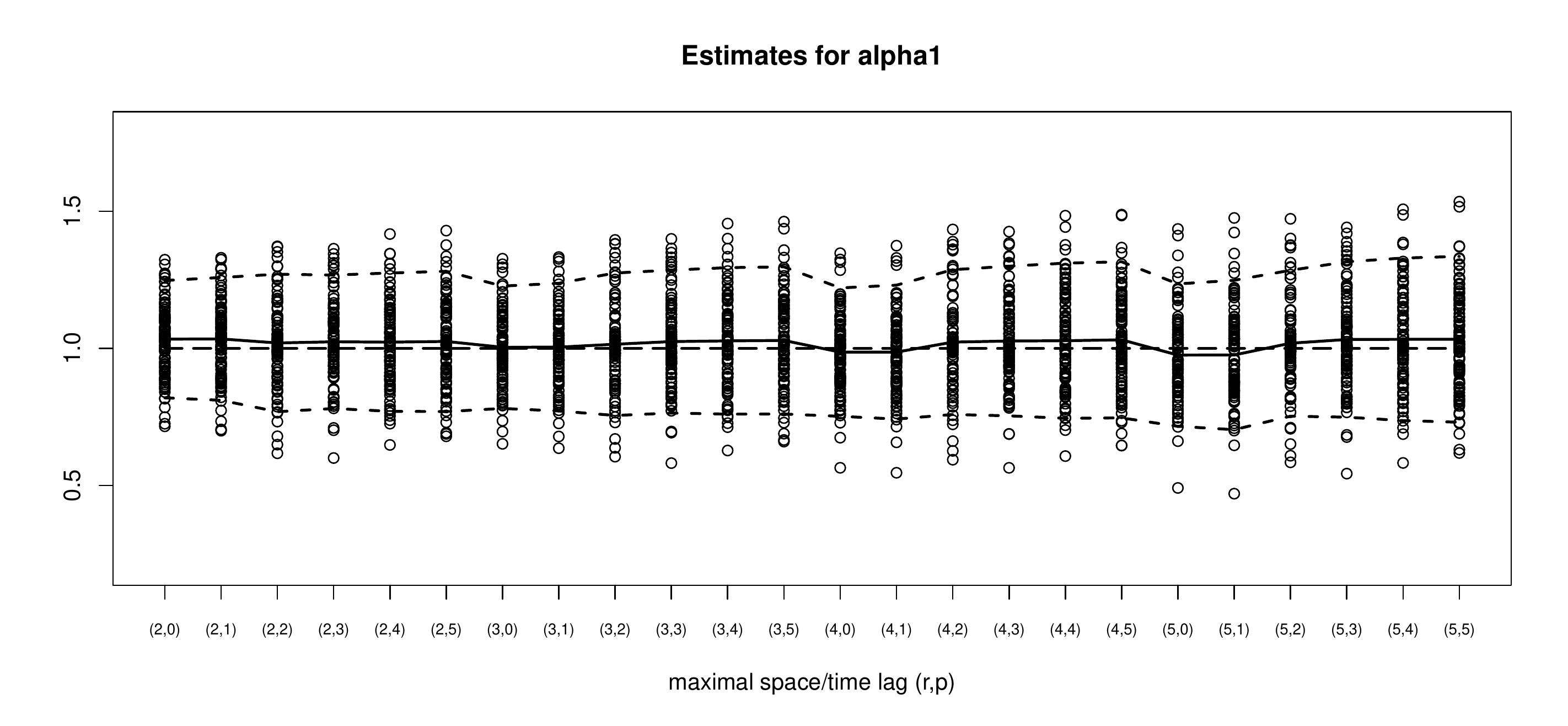}
\caption{Estimates for $\theta_1$ and $\alpha_1$ (spatial parameters) as a function of maximal space-time lags $(r,p)$. Each dot represents the estimate for one of the 100 simulations. The solid line is the mean over all the estimates for each fixed combination of $r$ and $p$ and the dotted lines are simulation-based 95\% pointwise confidence bands. The middle long dashed line represents the true value.}
\label{Sim1}
\end{figure}

\begin{figure}[htp]
\centering
\includegraphics[width=14.5cm, totalheight=8cm]{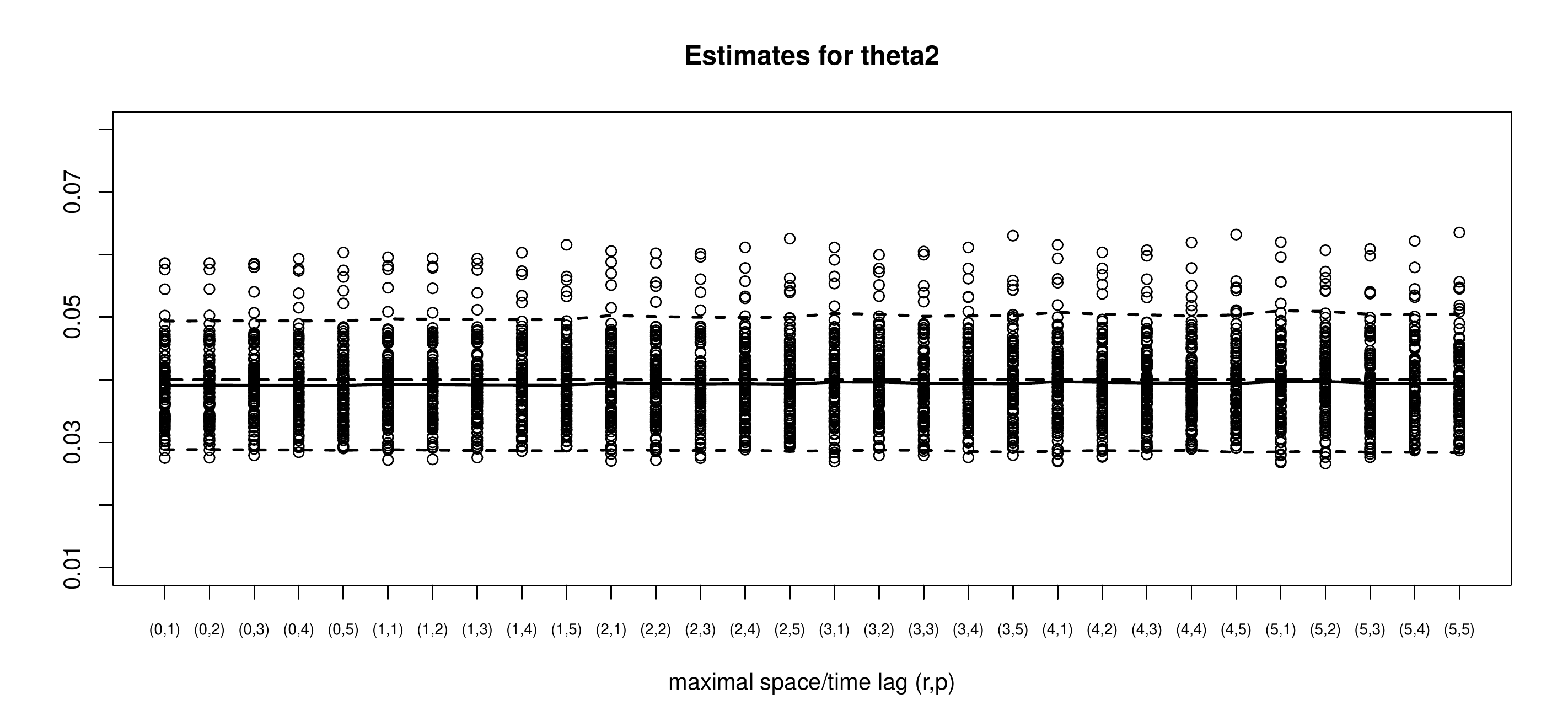}
\includegraphics[width=14.5cm, totalheight=8cm]{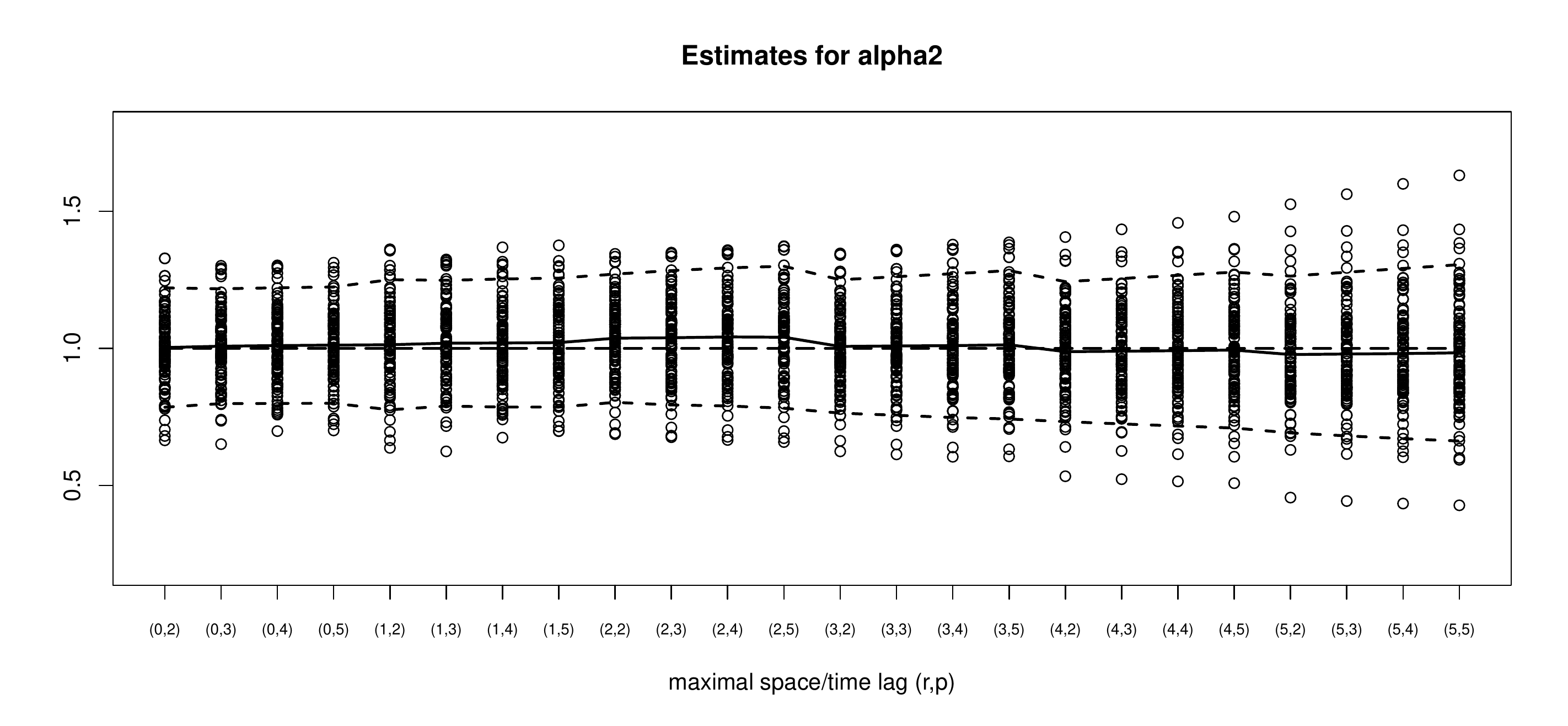}
\caption{Estimates for $\theta_2$ and $\alpha_2$ (temporal parameters) as a function of maximal space-time lags $(r,p)$. Each red dot represents the estimate for one of the 100 simulations. The solid line is the mean over all the estimates for each fixed combination of $r$ and $p$ and the dotted lines are simulation-based 95\% pointwise confidence bands. The middle long dashed represents the true value.}
\label{Sim2}
\end{figure}

\noindent
\textbf{Acknowledgments}\\
All authors gratefully acknowledge the support by the TUM Institute for Advanced Study (TUM-IAS).
The third author additionally likes to thank the International Graduate School of Science and Engineering (IGSSE) of the Technische Universit\"at M\"unchen for their support.

\bibliographystyle{plain}
\bibliography{bibtex_spacetime}

\begin{appendix}
\section{Proof of Lemma \ref{finiteder}}\label{appendix}
In the following, we use the same abbreviations as in Lemma \ref{bivdenlemma}.
The gradient of the bivariate log-density with respect to the parameter vector $\bs{\psi}$ is given by 
$$\nabla_{\bs{\psi}}\log f(x_1,x_2) = \frac{\partial \log f(x_1,x_2)}{\partial \delta}\nabla_{\bs{\psi}}\delta.$$
Assume in the following that all parameters $\theta_1,\alpha_1,\theta_2$ and $\alpha_2$ are identifiable.  
Since all partial derivatives 
$$\frac{\partial \delta}{\partial \theta_1} = \|\bs{h}\|^{\alpha_1}, \  \frac{\partial \delta}{\partial \theta_2} =|u|^{\alpha_2}, \   \frac{\partial \delta}{\partial \alpha_1} \theta_1\alpha_1\|\bs{h}\|^{\alpha_1-1}, \  \frac{\partial \delta}{\partial \alpha_2} = \theta_2\alpha_2|u|^{\alpha_2-1}, $$
as well as all second order partial derivatives 
can be bounded from below and above for $0<\min\left\{\|\bs{h}\|,|u|\right\},\max\left\{\|\bs{h}\|,|u|\right\}<\infty$ using assumption \eqref{parspace} and, independently of the parameters $\theta_1$, $\theta_2$, $\alpha_1$ and $\alpha_2$, it suffices to show that 
$$\mathbb{E}_{\bs{\psi}^*}\left[\left|\frac{\partial \log f_{\bs{\psi}}(\eta(\bs{s}_1,t_1),\eta(\bs{s}_2,t_2))}{\partial \delta}\right|^{3}\right] < \infty$$
and 
$$\mathbb{E}_{\bs{\psi}^*}\left[\sup\limits_{\bs{\psi}\in\bs{\Psi}}\left|\frac{\partial^2\log f_{\bs{\psi}}(\eta(\bs{s}_1,t_1),\eta(\bs{s}_2,t_2))}{\partial \delta}\right|\right]< \infty.$$
Since $\delta$ can be bounded away from zero using assumption \eqref{parspace}, we can treat $\delta$ as a constant. For simplification we drop the argument in the following equalities. 
Define 
$$A_1 = \frac{\partial V}{\partial x_1}, \ A_2 = \frac{\partial V}{\partial x_2}, \ \text{ and }  A_3 = \frac{\partial^2 V}{\partial x_1 x_2}.$$
The partial derivative of the bivariate log-density with respect to $\delta$ has the following form 
$$\frac{\partial \log f_{\bs{\psi}}}{\partial \delta} = -\frac{\partial V}{\partial \delta} + (A_1A_2-A_3)^{-1}\left(\frac{\partial A_1}{\partial \delta} A_2 + A_1\frac{\partial A_2}{\partial \delta} - \frac{\partial A_3}{\partial \delta}\right).$$
We identify stepwise the ``critical'' terms, where ``critical'' means higher order terms of functions of $x_1$ and $x_2$. 
To give an idea on how to handle the components in the derivatives, we describe one such step. 
Note that $(A_1A_2-A_3)^{-1}$ can be written as 
$$(A_1A_2-A_3)^{-1} = \frac{x_1x_2}{g_1\left(\frac{1}{x_1},\frac{1}{x_2},\frac{1}{x_1x_2},\frac{1}{x_1^2},\frac{1}{x_2^2}\right)},$$
where $g_1$ describes the sum of the components together with additional multiplicative factors. 
By using 
$$\frac{\partial \Phi(q_{\bs{\psi}}^{(1)})}{\partial \delta} = \frac{q_{\bs{\psi}}^{(1)}}{2\delta}\varphi(q_{\bs{\psi}}^{(1)}) \quad \text{ and } \quad  \frac{\partial \varphi(q_{\bs{\psi}}^{(1)})}{\partial \delta} = -\frac{(q_{\bs{\psi}}^{(1)})^2}{2\delta}\varphi(q_{\bs{\psi}}^{(1)}),$$
where $q_{\bs{\psi}}^{(1)} = \log(x_2/x_1)/(2\sqrt{\delta}) + \sqrt{\delta}$, we have 
$$\frac{\partial A_1}{\partial \delta} A_2 = g_2\left(\frac{1}{x_1^2,x_2^2},\frac{q_{\bs{\psi}}^{(1)}}{x_1^2x_2^2},\frac{(q_{\bs{\psi}}^{(1)})^2}{x_1^2x_2^2},\frac{1}{x_1^3x_2}, \frac{q_{\bs{\psi}}^{(1)}}{x_1^3x_2}, \frac{(q_{\bs{\psi}}^{(1)})^2}{x_1^3x_2},\frac{1}{x_1x_2^3},\frac{(q_{\bs{\psi}}^{(1)})^2}{x_1x_2^3}\right),$$
where $g_2$ is a linear function of the components. 
By combining the two representations above, we obtain that all terms in \\ $(A_1A_2-A_3)^{-1}(\partial A_1/\partial \delta) A_2$ are of the form
\begin{equation}
\frac{|\log x_1|^{k_1}|\log x_2|^{k_2}}{x_1^{k_3}x_2^{k_4}}, \ \ k_1,k_2,k_3,k_4 \geq 0. 
\label{critical}
\end{equation}
The second derivative of the bivariate log-density with respect to $\delta$ is given by 
\begin{align*}
\frac{\partial^2 \log f_{\bs{\psi}}}{(\partial \delta)^2} =& -\frac{\partial^2 V}{(\partial \delta)^2} - (A_1A_2-A_3)^{-2} \left(\frac{\partial A_1}{\partial \delta} A_2 + A_1\frac{\partial A_2}{\partial \delta} - \frac{\partial A_3}{\partial \delta}\right)^2 \\
&{} + (A_1A_2-A_3)^{-1}\left(\frac{\partial^2A_1}{(\partial \delta)^2}A_2 + 2\frac{\partial A_1}{\partial \delta}\frac{\partial A_2 }{\partial \delta} + A_1\frac{\partial^2A_2}{(\partial \delta)^2} -\frac{\partial^2 A_3}{(\partial \delta)^2}\right)
\end{align*}
Stepwise calculation of the single components shows that all terms are also of form \eqref{critical}. This implies that for both statements it suffices to show that for all $k_1,k_2,k_3,k_4\geq 0$
$$\mathbb{E}\left[\frac{(\log\eta(\bs{s},t))^{k_1}(\log\eta(\bs{s},t))^{k_2}}{|\eta(\bs{s},t)|^{k_3}|\eta(\bs{s},t)|^{k_4}}\right] < \infty.$$
Since $\eta(\bs{s},t)$ is standard Fr\'echet $\log(\eta(\bs{s},t))$ is standard Gumbel and $1/\eta(\bs{s},t)$ is standard exponential. Using H\"older's inequality, we obtain 
\begin{align*}
&\mathbb{E}\left[\frac{|\log(\eta(\bs{s},t))|^{k_1}|\log(\eta(\bs{s},t))|^{k_2}}{|\eta(\bs{s},t)|^{k_3}|\eta(\bs{s},t)|^{k_4}}\right] \\
&< \left(\mathbb{E}\left[|\log(\eta(\bs{s},t))|^{4k_1}\right]\mathbb{E}\left[|\log(\eta(\bs{s},t))|^{4k_2}\right]\right)^{1/2}
\left(\mathbb{E}\left[\left|\frac{1}{\eta(\bs{s},t)}\right|^{4k_3}\right]\mathbb{E}\left[\left|\frac{1}{\eta(\bs{s},t)}\right|^{4k_4}\right]\right)^{1/2} <\infty,
\end{align*} 
since all moments of the exponential and the Gumbel distributions are finite. 

\end{appendix}

\end{document}